\documentclass[aps, prx, a4paper, floatfix, twocolumn, longbibliography, footinbib, amssymb]{revtex4-2}

\usepackage[colorlinks=true, hypertexnames=false]{hyperref}
\PassOptionsToPackage{linktocpage}{hyperref} 
\hypersetup{
  colorlinks  = true, 
  urlcolor = magenta!80!black,
  linkcolor  = blue!60!black, 
  citecolor  = blue!60 
}

\usepackage[T1]{fontenc}
\usepackage[utf8]{inputenc}
\usepackage{epsfig}
\usepackage[english]{babel}
\usepackage[table]{xcolor}
\usepackage{graphicx}
\usepackage{amsmath}
\usepackage{relsize}
\usepackage[percent]{overpic}
\usepackage{bm}
\usepackage{hyperref}
\usepackage{siunitx}
\usepackage{accents}
\usepackage{mathtools}
\usepackage{grffile}
\usepackage[normalem]{ulem}
\usepackage{pgf}
\usepackage{tikz}
\usetikzlibrary{graphs}
\usepackage{times}
\usepackage[capitalize]{cleveref}
\usepackage{dsfont}  

\newcommand{\avg}[1]{{\left<#1\right>}}

\newcommand{\ceil}[1]{{\lceil #1\rceil}}
\newcommand{\dd}{\mathrm{d}}
\newcommand{\ee}{\mathrm{e}}

\usepackage{mathtools}
\usepackage{amsthm}

\def\multiset#1#2{\ensuremath{\left(\kern-.3em\left(\genfrac{}{}{0pt}{}{#1}{#2}\right)\kern-.3em\right)}}

\usepackage{amsmath}

\theoremstyle{plain}
\newtheorem{theorem}{Theorem}

\newtheorem{lemma}[theorem]{Lemma}

\theoremstyle{definition}

\theoremstyle{remark}

\newcommand{\A}{\bm{A}}
\newcommand{\G}{\bm{G}}

\newcommand{\W}{\bm{W}}

\newcommand{\argmax}{\operatorname{arg\ max}}

\newcommand{\etal}{\textit{et al}}

\newcolumntype{P}[1]{>{\centering\arraybackslash}p{#1}}

\usepackage{verbatim}
\usepackage{overpic}
\usepackage{booktabs}
\usepackage{placeins}

\usepackage{siunitx}
\usepackage{multirow}
\usepackage{algpseudocodex}
\usepackage{algorithm}

\algnewcommand\Input{\State\textbf{Input:} }
\algrenewcommand\Output{\State\textbf{Output:} }
\algnewcommand\Continue{\State\textbf{continue}}

\begin{document}

\title{Scalable network reconstruction in subquadratic time}

\author{Tiago \surname{P. Peixoto}}
\email{peixotot@ceu.edu}
\affiliation{Department of Network and Data Science, Central European University, Vienna, Austria}

\begin{abstract}
  Network reconstruction consists in determining the unobserved pairwise
  couplings between $N$ nodes given only observational data on the resulting
  behavior that is conditioned on those couplings---typically a time-series or
  independent samples from a graphical model. A major obstacle to the
  scalability of algorithms proposed for this problem is a seemingly unavoidable
  quadratic complexity of $\Omega(N^2)$, corresponding to the requirement of
  each possible pairwise coupling being contemplated at least once, despite the
  fact that most networks of interest are sparse, with a number of non-zero
  couplings that is only $O(N)$. Here we present a general algorithm applicable
  to a broad range of reconstruction problems that significantly outperforms
  this quadratic baseline. Our algorithm relies on a stochastic second neighbor
  search~\cite{dong_efficient_2011} that produces the best edge candidates with
  high probability, thus bypassing an exhaustive quadratic search. If we rely on
  the conjecture that the second-neighbor search finishes in log-linear
  time~\cite{baron_k-nearest_2020,baron_empirical_2022}, we demonstrate
  theoretically that our algorithm finishes in subquadratic time, with a
  data-dependent complexity loosely upper bounded by $O(N^{3/2}\log N)$, but
  with a more typical log-linear complexity of $O(N\log^2N)$. In practice, we
  show that our algorithm achieves a performance that is many orders of
  magnitude faster than the quadratic baseline---in a manner
  consistent with our theoretical analysis---allows for easy parallelization, and
  thus enables the reconstruction of networks with hundreds of thousands and
  even millions of nodes and edges.
\end{abstract}

\maketitle

\section{Introduction}

Networks encode the pairwise interactions that determine the dynamical behavior
of a broad class of interconnected systems. However, in many important cases
the interactions themselves are not directly observed, and instead we have
access only to their indirect outcomes, usually as samples from a multivariate
distribution modeled as a probabilistic graphical
model~\cite{lauritzen_graphical_1996,jordan_graphical_2004,drton_structure_2017},
or from time-series data representing some dynamics conditioned on the network
structure~\cite{timme_revealing_2014,hallac_network_2017}. Instances of this
problem include the inference of interactions between microbial species from
co-occurrence data~\cite{guseva_diversity_2022}, financial market couplings from
stock prices~\cite{bury_statistical_2013}, protein structure from amino-acid
contact maps~\cite{weigt_identification_2009}, gene regulatory networks from
expression data~\cite{dhaeseleer_genetic_2000}, neural connectivity from fMRI
and EEG data~\cite{manning_topographic_2014}, and epidemic contact
tracing~\cite{braunstein_alfredo_network_2019}, among others.

Perhaps the most well studied formulation of the network reconstruction problem
is covariance selection~\cite{dempster_covariance_1972}, where it is assumed
that the data consist of independent samples of a multivariate Gaussian, and the
objective is to infer its precision matrix --- often assumed to be sparse. The
most widely employed algorithm for this purpose is the graphical LASSO
(GLASSO)~\cite{friedman_sparse_2008}, and its many
variations~\cite{mazumder_graphical_2012,hastie_statistical_2015}. More
generally, one can consider arbitrary probabilistic graphical models (a.k.a.
Markov random fields)~\cite{bresler_reconstruction_2008}, where the latent
network structure encodes the conditional dependence between variables.
Covariance selection is a special case of this family where the variables are
conditionally normally distributed, resulting in an analytical likelihood,
unlike the general case which involves intractable normalization constants. A
prominent non-Gaussian graphical model is the Ising
model~\cite{nguyen_inverse_2017}, applicable for binary variables, which has a
wide range of applications.

The vast majority of algorithmic approaches so far employed to the network
reconstruction problem cannot escape a complexity of at least $O(N^2)$, where
$N$ is the number of nodes in the network. The original GLASSO algorithm for
covariance selection has a complexity of $O(N^3)$. By exploiting properties that
are specific to the covariance selection problem (and hence do not generalize to
the broader reconstruction context), the faster QUIC~\cite{hsieh_quic_2014} and
\textsc{BigQUIC}~\cite{hsieh_big_2013} approximative methods have $O(N^2)$ and
$O(NE)$ complexities, respectively, with $E$ being the number of edges (i.e.
nonzero entries in the reconstructed matrix), such that the latter also becomes
quadratic in the usual sparse regime with $E=O(N)$. Likewise, for the inverse
Ising model~\cite{nguyen_inverse_2017,bresler_efficiently_2015} or graphical
models in general~\cite{bresler_reconstruction_2008,bresler_reconstruction_2010}
no known method can improve on a $O(N^2)$ complexity, and the same is true for
reconstruction from
time-series~\cite{bresler_learning_2018,timme_revealing_2014}. To the best of
our knowledge, no general approach exists to the network reconstruction problem
with a lower complexity than $O(N^2)$, unless strong assumptions on the true
network structure are made. Naively, one could expect this barrier to be a
fundamental one, since for the reconstruction task --- at least in the general
case --- we would be required to probe the existence of every possible pairwise
coupling at least once.

Instead, in this work we show that it is in fact possible to implement a general
network reconstruction scheme that yields subquadratic complexity, without
relying on the specific properties of any particular instance of the problem.
Our approach is simple, and relies on a stochastic search for the best update
candidates (i.e.\ edges that need to be added, removed, or updated) in an
iterative manner that starts from a random graph and updates the candidate list
by inspecting the second neighbors of this graph --- an approach which leads to
log-linear
performance~\cite{dong_efficient_2011,baron_k-nearest_2020,baron_empirical_2022}.
Furthermore, every step of our algorithm is easily parallelizable, allowing its
application for problems of massive size.

This paper is organized as follows. In Sec.~\ref{sec:setup} we introduce the
general reconstruction scenario, and the coordinate descent algorithm, which will
function as our baseline with quadratic complexity. In Sec.~\ref{sec:gcd} we
describe our improvement over the baseline, and analyze its algorithmic
complexity. In Sec.~\ref{sec:performance} we evaluate the performance of our
approach on a variety of synthetic and empirical data, and in
Sec.~\ref{sec:empirical} we showcase our algorithm with some selected
large-scale empirical network reconstruction problems. We finalize in
Sec.~\ref{sec:conclusion} with a conclusion.

\section{General reconstruction scenario and the coordinate descent (CD) baseline}\label{sec:setup}

We are interested in a general reconstruction setting where we assume some data
$\bm X$ are sampled from a generative model with a likelihood
\begin{equation}
  P(\bm X | \W),
\end{equation}
where $\W$ is a $N\times N$ symmetric matrix corresponding to the weights of an
undirected graph of $N$ nodes. In most cases of interest, the matrix $\W$ is
sparse, i.e.\ its number of non-zero entries is $O(N)$, but otherwise we assume
no special structure. Typically, the data $\bm X$ are either a $N\times M$
matrix of $M$ independent samples, with $X_{im}$ being a value associated with
node $i$ for sample $m$, such that
\begin{equation}
  P(\bm X | \W) = \prod_{m=1}^{M}P(\bm x_{m} | \W),
\end{equation}
with $\bm x_{m}$ being the $m$-th column of $\bm X$, or a Markovian time series with
\begin{equation}
  P(\bm X | \W) = \prod_{m=1}^{M}P(\bm x_{m} | \bm x_{m-1},\W),
\end{equation}
given some initial state $\bm x_{0}$. Our algorithm will not rely strictly on
any such particular formulations, only on a generic posterior distribution
\begin{equation}
  \pi(\W) = P(\W | \bm X) = \frac{P(\bm X| \W)P(\W)}{P(\bm X)}
\end{equation}
that needs to be computable only up to normalization. We focus on the MAP point
estimate
\begin{equation}\label{eq:MAP}
  \widehat{\W} = \underset{\W}{\argmax}\;\pi(\W).
\end{equation}
For many important instances of this problem, such as covariance
selection~\cite{dempster_covariance_1972,friedman_sparse_2008} and the inverse
Ising model~\cite{nguyen_inverse_2017} the optimization objective above is
convex. In this case, one feasible approach is the coordinate descent (CD)
algorithm~\cite{tseng_convergence_2001,wright_coordinate_2015}, which proceeds
by iterating over all variables in sequence, and solving a one-dimensional
optimization (which is guaranteed to be convex as well), and stopping when a
convergence threshold is reached (see Algorithm~\ref{alg:cd}).

Algorithm~\ref{alg:cd} has complexity $O(\tau N^2)$, assuming step (1),
corresponding to an element-wise optimization, can be done in time $O(1)$ (e.g.\
using bisection search), where $\tau$ is the number of iterations required for
convergence --- which in general will depend on the particulars of the problem
and initial state, but typically we have $\tau = O(1)$. We observe that in all
our analyses we will assume that the initial state $\W_{0}$ corresponds to an
entirely empty network with every entry being equal to zero---corresponding to
the worst-case scenario of maximum initial ignorance about the network to be
reconstructed.

\begin{algorithm}[H]
  \caption{Coordinate descent (CD)}\label{alg:cd}
  \begin{minipage}{.96\columnwidth}
    \begin{algorithmic}
      \Input Objective $\pi(\W)$, initial state $\W_{0}$, convergence criterion $\epsilon$
      \Output Estimate $\widehat{\W} = \underset{\W}{\argmax}\;\pi(\W)$
      \State $\W \gets \W_0$
      \Repeat
      \State $\Delta \gets 0$
      \ForAll {$i < j$}
      \State $W_{ij}' \gets \argmax_{W_{ij}}\; \pi (\W)$ \Comment{(1)}
      \State $\Delta \gets \Delta + |W_{ij}'- W_{ij}|$
      \State $W_{ij} \gets W_{ij}'$
      \EndFor
      \Until {$\Delta < \epsilon$}
      \State $\widehat\W \gets \W$
    \end{algorithmic}
  \end{minipage}
\end{algorithm}

Note that the CD algorithm does not require a differentiable objective
$\pi(\W)$, but convergence to the global optimum is only guaranteed if it is
convex and sufficiently smooth~\cite{spall_cyclic_2012}. In practice, CD is the
method of choice for covariance selection and the inverse Ising model, with a
speed of convergence that often exceeds gradient descent (which is not even
strictly applicable when non-differentiable regularization is used, such as the
$L_{1}$ of GLASSO~\cite{friedman_sparse_2008}), since each coordinate can
advance further with more independence from the remaining ones, unlike with
gradient descent where all coordinates are restricted by the advance of the
slowest one.

For a nonconvex objective $\pi(\W)$ the CD algorithm will in general not
converge to the global optimum. Nevertheless, it is a fundamental baseline that
often gives good results in practice even in nonconvex instances, and can serve
as a starting point for more advanced algorithms. In this work we are not
primarily concerned with issues due to nonconvexity, but rather with a general
approach that circumvents the need to update all ${N\choose 2}$ entries of the
matrix $\W$.

Our objective is to reduce the $O(N^{2})$ complexity of the CD algorithm. But,
before continuing, we will remark on the feasibility of the reconstruction
problem, and the obstacle that this quadratic complexity represents. At first,
one could hypothesize that the size of the data matrix $\bm X$ would need to be
impractically large to allow for the reconstruction of networks with $N$ in the
order of hundreds of thousands or millions. In such a case, a quadratic
complexity would be the least of our concerns for problem sizes that are
realistically feasible. However, for graphical models it is possible to show
that the number of samples required for accurate reconstruction scales only with
$M=O(\log N)$~\cite{bresler_reconstruction_2008,abbeel_learning_2006,wainwright_high-dimensional_2006,santhanam_information-theoretic_2009,bresler_reconstruction_2010},
meaning that reconstruction of large networks with relatively little information
is possible. In view of this, a quadratic algorithmic complexity on $N$ poses a
significant obstacle for practical instances of the problem, which could easily
become more limited by the runtime of the algorithm than the available data.

\section{Subquadratic network reconstruction}\label{sec:gcd}

Our algorithm is based on a greedy extension of the CD algorithm~\ref{alg:cd}
(GCD), where we select only the $\kappa N$ entries of the matrix $\W$ that would
individually lead to the steepest increase of the objective function $\pi(\W)$,
as summarized in Algorithm~\ref{alg:gcd}.

\begin{algorithm}[H]
\caption{Greedy coordinate descent (GCD)}\label{alg:gcd}
\begin{minipage}{.96\columnwidth}
  \begin{algorithmic}
    \Input Objective $\pi(\W)$, greediness factor $\kappa$, initial state $\W_{0}$, convergence criterion $\epsilon$
    \Output Estimate $\widehat{\W} = \underset{\W}{\argmax}\;\pi(\W)$

    \State $\W \gets \W_0$
    \Repeat
    \State $\Delta \gets 0$
    \State $\mathcal{E}_{\text{best}}\gets \text{\textsc{FindBest}}(\lfloor\kappa N\rceil, \{1,\dots,N\}, \textsc{d})$ \Comment{$|\mathcal{E}_{\text{best}}| = \lfloor\kappa N\rceil$}
    \ForAll {$(i,j) \in \mathcal{E}_{\text{best}}$}
    \State $W_{ij}' \gets \argmax_{W_{ij}}\; \pi (\W)$
    \State $\Delta \gets \Delta + |W_{ij}'- W_{ij}|$
    \State $W_{ij} \gets W_{ij}'$
    \EndFor
    \Until {$\Delta < \epsilon$}
    \State $\widehat\W \gets \W$\\

    \Function{d}{$i,j$} \Comment{``Distance'' function}
    \State
    \Return $-\max_{W_{ij}} \pi(\W)$
    \EndFunction\\
  \end{algorithmic}
\end{minipage}
\end{algorithm}

\begin{figure}
  \includegraphics[width=.7\columnwidth]{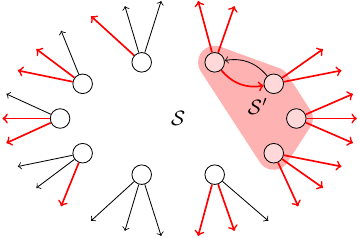}
  \caption{Diagrammatic representation of the sets $\mathcal S$ and
    $\mathcal{S}'$ in algorithm~\ref{alg:mclosest} for $k=3$. Edges marked in
    red belong to set $\mathcal D^+$, i.e. the $2m$ directed pairs $(i,j)$ with
    smallest $\text{\textsc{d}}(i,j)$. Note that reciprocal edges need not both
    belong to $\mathcal D^+$, despite $\text{\textsc{d}}(i,j)$ being symmetric,
    since ties are broken arbitrarily. The nodes in red have all out-edges
    (nearest neighbors) in set $\mathcal D$, and hence are assigned to set
    $\mathcal{S}'$. Since the set of $m$ best pairs could still contain
    undiscovered pairs of elements in $\mathcal{S}'$, the search needs to
    continue recursively for members of this set.\label{fig:Sdiagram}}
\end{figure}

In the above algorithm, the function
$\text{\textsc{FindBest}}(m,\mathcal S,\text{\textsc{d}})$ finds the set of $m$
pairs $(i,j)$ of elements in set $\mathcal S$ with the smallest ``distance''
$\textsc{d}(i,j)$ [which in our case corresponds to $-\max_{W_{ij}} \pi(\W)$].
Clearly, for $\lfloor\kappa N\rceil > 0$, if the original CD algorithm converges
to the global optimum, so will the GCD algorithm. The function \textsc{FindBest}
solves what is known as the $m$ closest pairs
problem~\cite{smid_closest-point_2000}. In that literature, $\textsc{d}(i,j)$ is
often assumed to be a metric, typically Euclidean, which allows the problem to
be solved in log-linear time, usually by means of spacial sorting. However, this
class of solution is not applicable to our case, since we cannot expect that our
distance function will in general define a metric space. An exhaustive solution
of this problem consists in probing all $|\mathcal S|\choose 2$ pairs, which
would yield no improvement on the quadratic complexity of CD. Instead, we
proceed by first solving the $m$ closest pairs problem with an algorithm
proposed by Lenhof and Smid~\cite{lenhof_k_1992} that maps it to a recursive
$k$-nearest neighbor (KNN) problem. The algorithm proceeds by setting
$k = \ceil{4m/|\mathcal{S}|}$ and finding for each element $i$ in set
$\mathcal S$ the $k$ nearest neighbors $j$ with smallest $\textsc{d}(i,j)$. From
this set of directed pairs, we select the $2m$ best ones to compose the set
$\mathcal D^+$, and construct a set $\mathcal D$ with the undirected version of
the pairs in $\mathcal D^{+}$, such that $m \le |\mathcal D| \le 2m$. At this
point we can identify a subset $\mathcal S'$ of $\mathcal S$ composed of nodes
for which all nearest neighbor edges belong to $\mathcal D$ when their direction
is discarded, as shown in Fig.~\ref{fig:Sdiagram}. Since these nodes have been
saturated, we cannot rule out that the $m$ closest pairs will not contain
undiscovered pairs of elements in set $\mathcal S'$. Therefore, we proceed
recursively for $\mathcal S'$, and stop when $|\mathcal S|^2 \le 4 m$, in which
case the node set has become small enough for an exhaustive search to be
performed. This is summarized as algorithm~\ref{alg:mclosest}, and a proof of
correctness is given in Ref.~\cite{lenhof_k_1992}.

\begin{algorithm}[H]
\caption{Find the $m$ best edge candidates.}\label{alg:mclosest}
\begin{minipage}{.96\columnwidth}
  \begin{algorithmic}
    \Function{FindBest}{$m, \mathcal{S},$ \textsc{d}}
    \If {$|\mathcal S|^2 \le 4 m$}
    \State
    \Return $\{m$ pairs $(i,j)$ of nodes in set $\mathcal{S}$ with smallest $\textsc{d}(i,j)$ found by exhaustive search.$\}$ \Comment{$O(|\mathcal{S}|^{2})$}
    \EndIf
    \State $k \gets \ceil{4m/|\mathcal{S}|}$
    \State $\bm G \gets \text{FindKNN}(k, \mathcal{S}, \text{\textsc{d}})$ \Comment{$k$-nearest neighbor digraph}
    \State $\mathcal D^+ \gets 2m$ directed edges $(i, j) \in \bm G$ with smallest $\text{\textsc{d}}(i,j)$.
    \State $\mathcal D \gets $ unique undirected pairs $(i,j)$  in $\mathcal D^{+}$.   \Comment {$m \le |\mathcal D| \le 2m$}
    \State $\mathcal{S}' \gets \{ i \in \mathcal{S}\, |\, (i, j) \in D, \forall \text{ out-neighbor $j$ of $i$ in $\bm G$.}\}$
    \State $\mathcal D' \gets $ \textsc{FindBest}($m, \mathcal{S}',$ \textsc{d})\\
    \Return $\{\text{$m$ pairs $(i,j)$ in $\mathcal D\cup \mathcal D'$ with smallest \textsc{d}$(i,j)$}\}$
    \EndFunction
  \end{algorithmic}
\end{minipage}
\end{algorithm}

As we will discuss in a moment, the recursion depth of
algorithm~\ref{alg:mclosest} is bounded logarithmically on $|\mathcal S|$, and
hence its runtime is dominated by the KNN search. Note that so far we have done
nothing substantial to address the overall quadratic performance, since finding
the $k$ nearest neighbors exhaustively still requires all pairs to be probed.
Similarly to the $m$ closest pairs problem, efficient log-linear algorithms
exist based on spacial sorting when the distance is Euclidean, however more
general approaches do also exist. In particular, the NNDescent algorithm by
Dong~\etal~\cite{dong_efficient_2011} approximately solves the KNN problem in
subquadratic time, while not requiring the distance function to be a metric. The
algorithm is elegant, and requires no special data structure beyond the nearest
neighbor graph itself, other than a heap for each node. It works by starting
with a random KNN digraph, and successively updating the list of best neighbors
by inspecting the neighbors of the neighbors in the undirected version of the
KNN digraph, as summarized in algorithm~\ref{alg:knn}. The main intuition behind
this approach is that if $(i,j)$ and $(j,v)$ are good entries to update, then
$(i,v)$ is likely to be a good candidate as well --- even if triangle inequality
is not actually obeyed.

\begin{algorithm}[H]
\caption{Find $k$ nearest neighbors by NNDescent.}\label{alg:knn}
\begin{minipage}{.96\columnwidth}
    \begin{algorithmic}
      \Input Convergence criterion $\varepsilon$
      \Function{FindKNN}{$k, \mathcal{V},$ \textsc{d}}
      \State $\bm G \gets $ directed graph with node set $\mathcal{V}$ and $k$ out-neighbors chosen uniformly at random independently for all nodes.
      \Repeat
      \State $\Delta \gets 0$
      \State $\G' \gets \G$
      \State $\bm U \gets$ undirected version of $\G$
      \ForAll{$i \in \mathcal{V}$}
          \ForAll{$j$ incident on $i$ in $\bm U$}
              \ForAll{$v$ incident on $j$ in $\bm U$}
                  \If{$v = i$ or $(i, v) \in \bm G'$}
                  \Continue
                  \EndIf
                  \State $\hat u \gets \argmax_{u} \{\text{\textsc{d}}(i, u)\, |\, (i, u) \in \bm G'\}$ \Comment{(1)}
                  \If{$\text{\textsc{d}}(i, v) < \text{\textsc{d}}(i, \hat u)$}
                    \State Replace $(i,\hat u)$ with $(i,v)$ in $\bm G'$ \Comment{(2)}
                    \State $\Delta \gets \Delta + 1$
                  \EndIf
              \EndFor
          \EndFor
      \EndFor
      \State $\G \gets \G'$
      \Until $\Delta / (k|\mathcal{V}|) < \varepsilon$
      \State \Return $\G$
      \EndFunction
    \end{algorithmic}
  \end{minipage}
\end{algorithm}

\begin{figure}
  \resizebox{\columnwidth}{!}{
  \begin{tabular}{ccc}
    (a) Random KNN graph & (b) Result of NNDescent & (c) Exact KNN graph\\
    \includegraphics[width=.33\columnwidth]{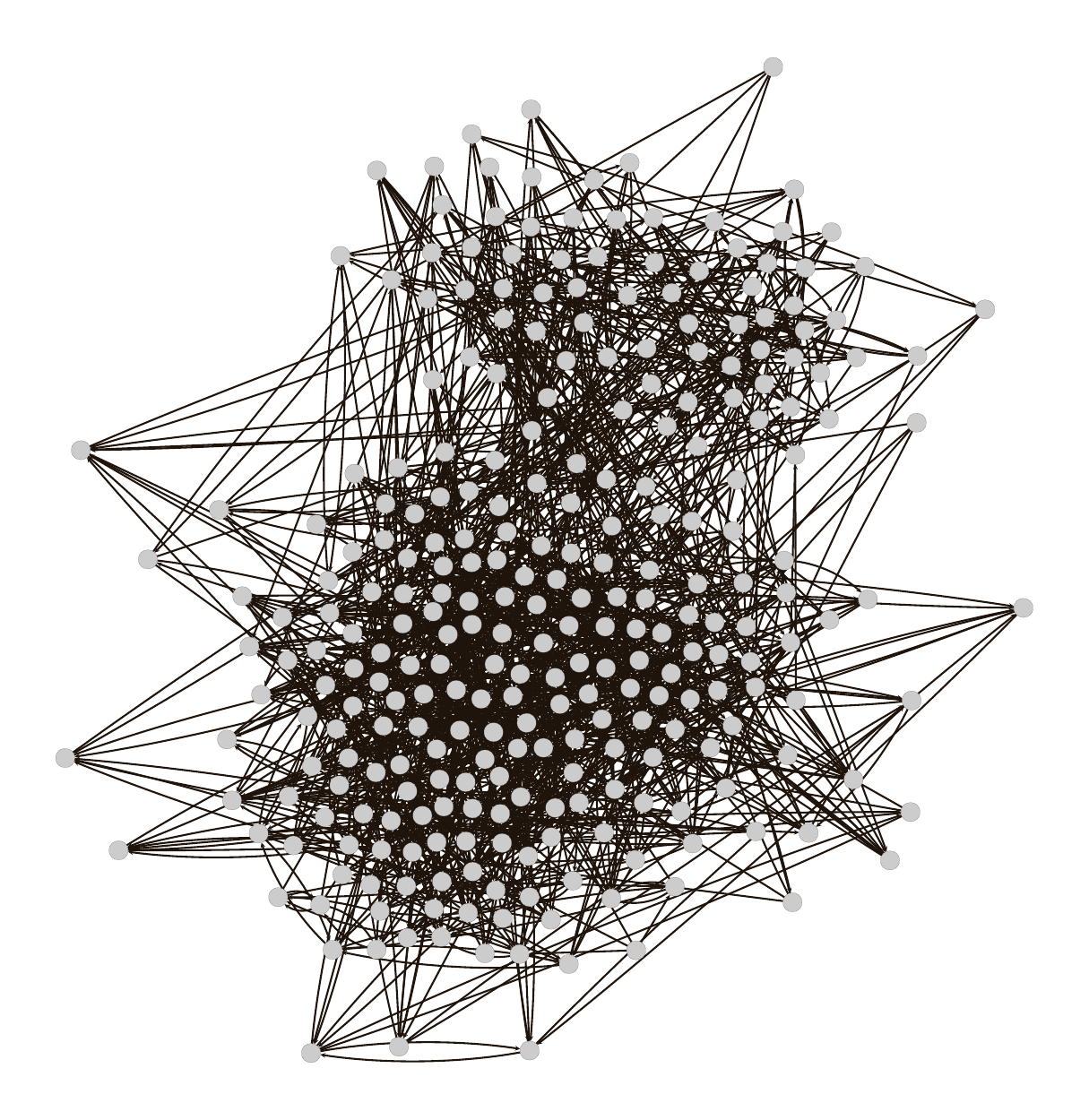} &
    \includegraphics[width=.33\columnwidth]{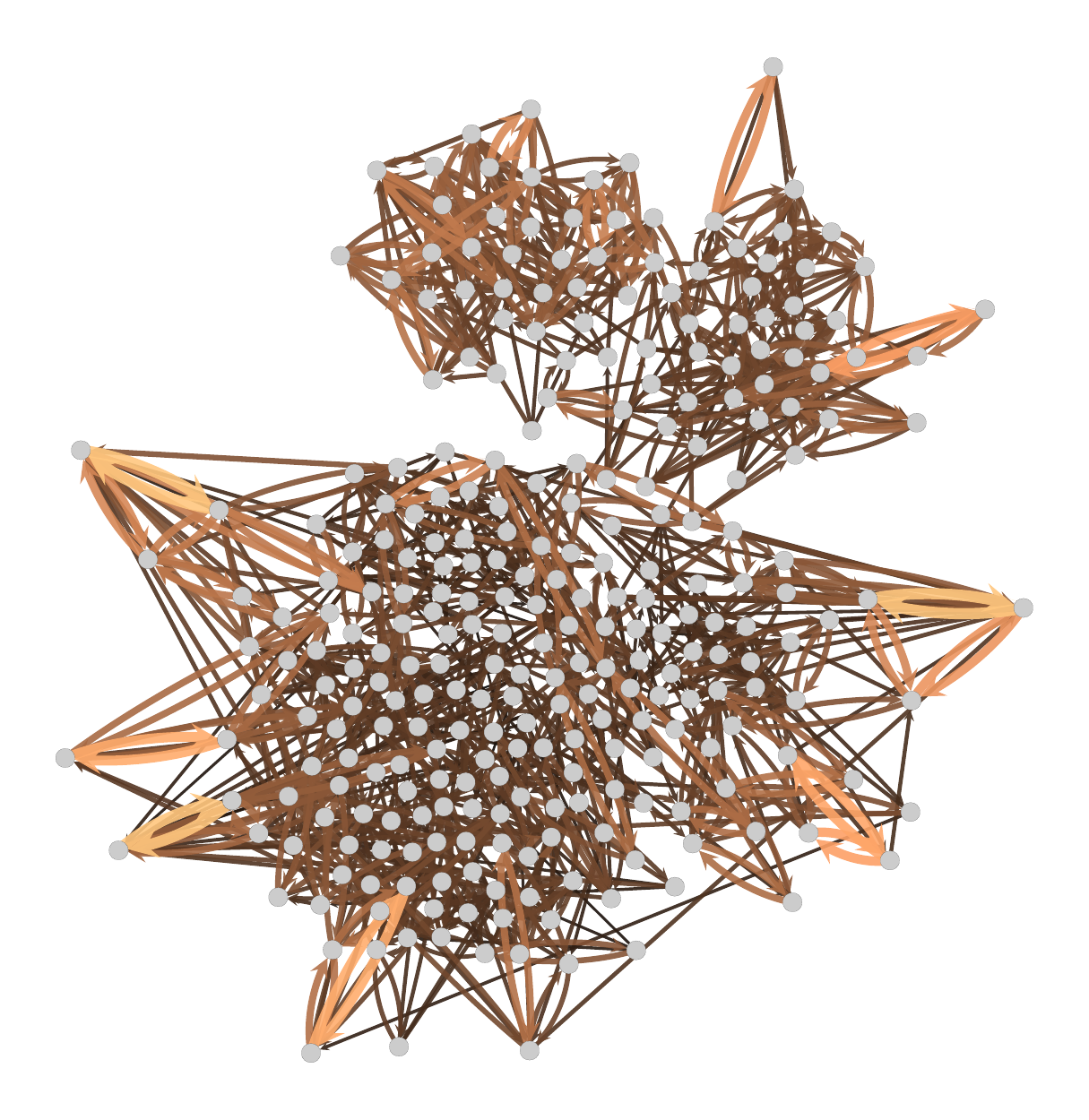} &
    \includegraphics[width=.33\columnwidth]{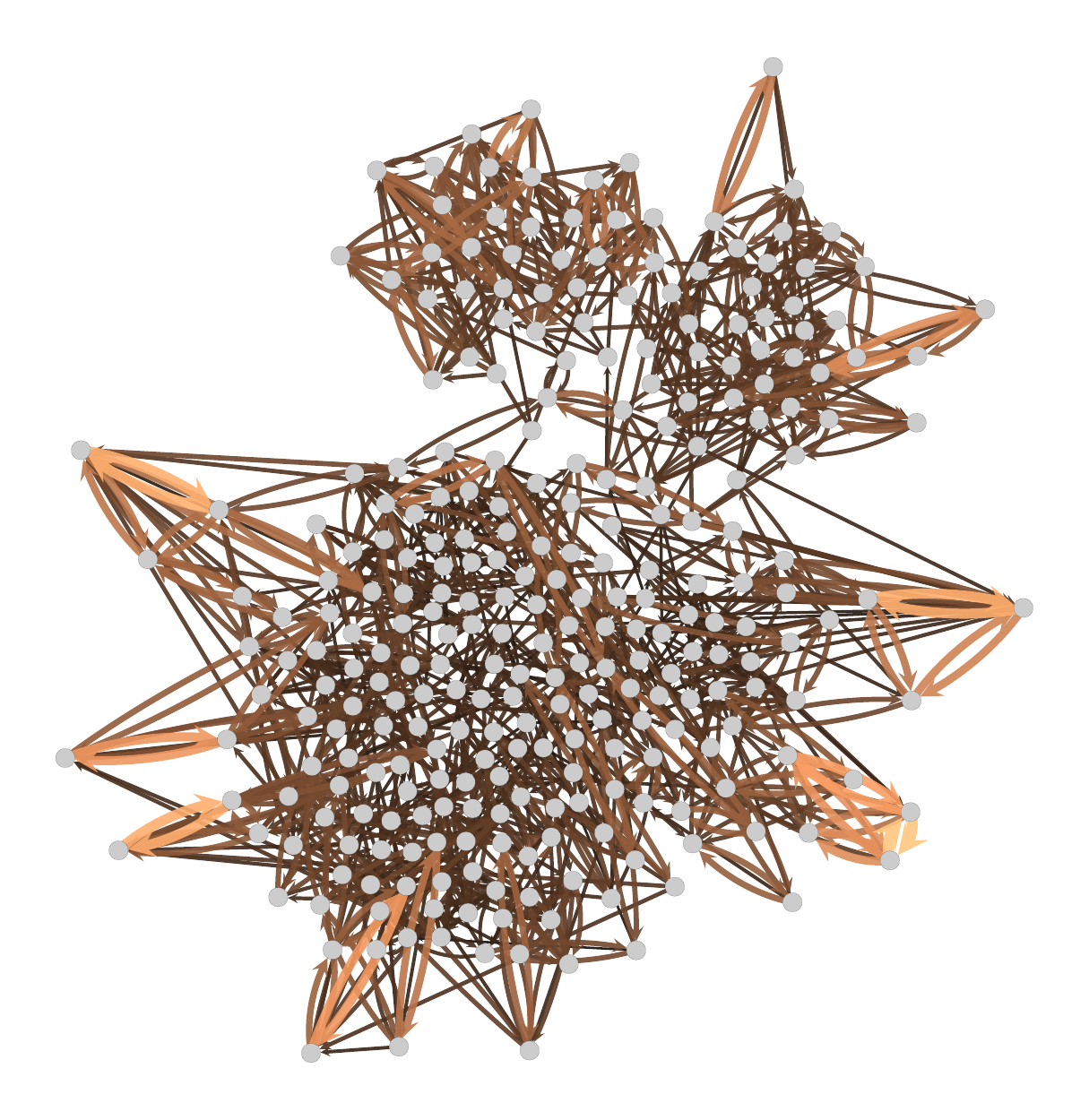}
  \end{tabular}}
  \resizebox{\columnwidth}{!}{
  \begin{tabular}{cc}
    (d) Best $m$ updates via \textsc{FindBest} & (e) Exact best $m$ updates \\
    \includegraphics[width=.5\columnwidth]{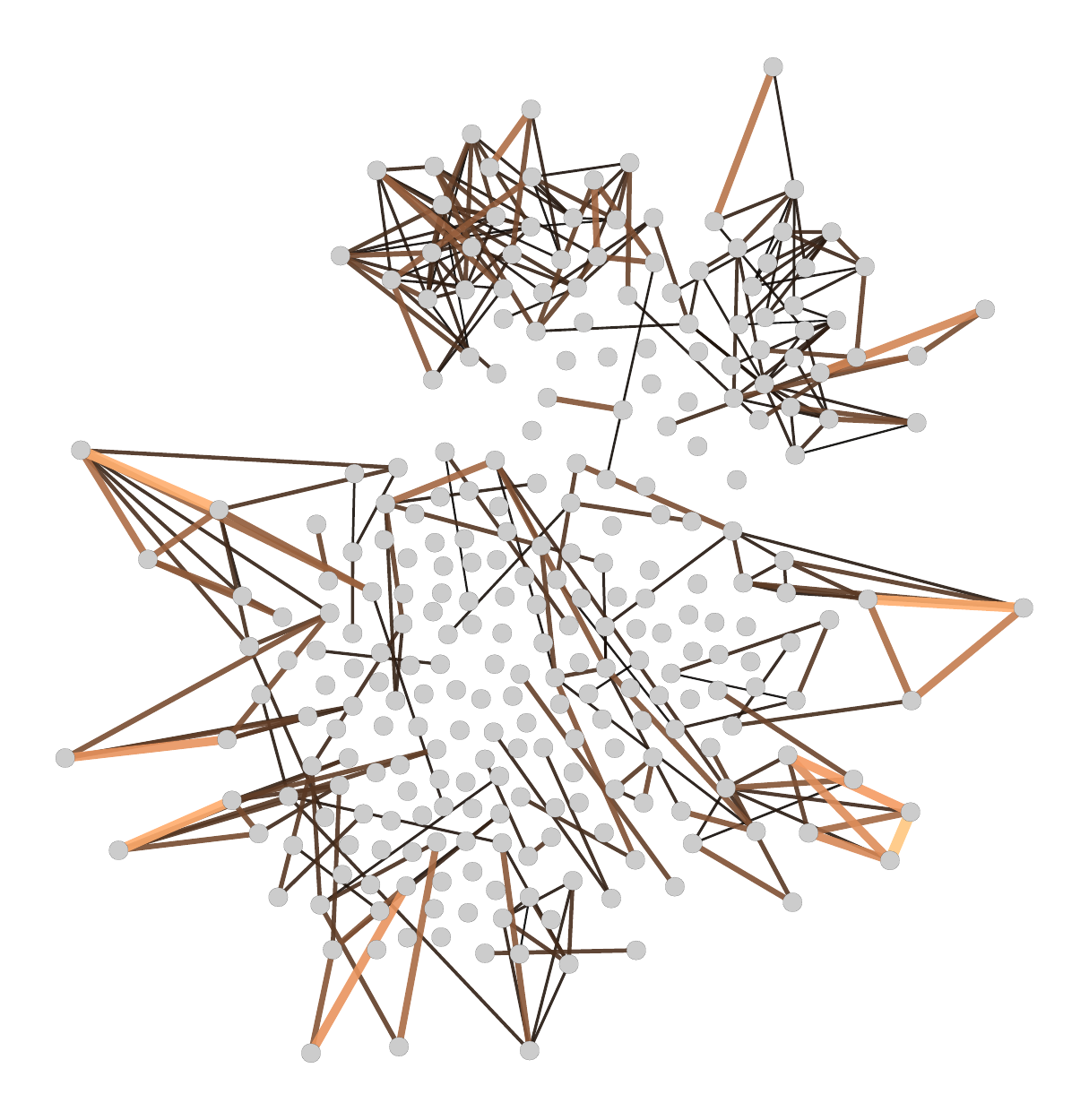} &
    \includegraphics[width=.5\columnwidth]{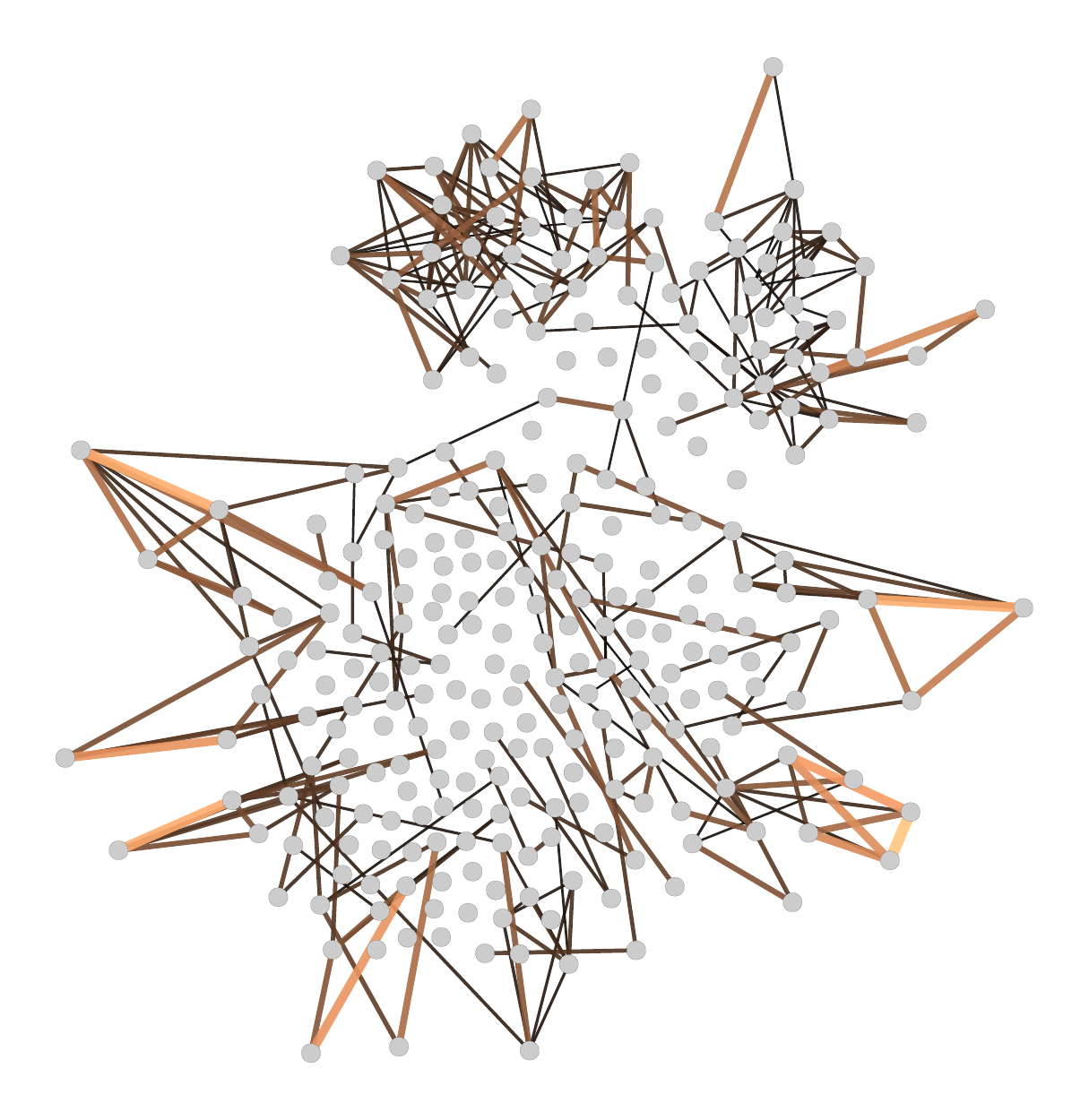}
  \end{tabular}}
  \resizebox{\columnwidth}{!}{
  \begin{tabular}{cccc}
    \multicolumn{4}{c}{(f) Iterations of the GCD algorithm}\\
    \includegraphics[width=.25\columnwidth]{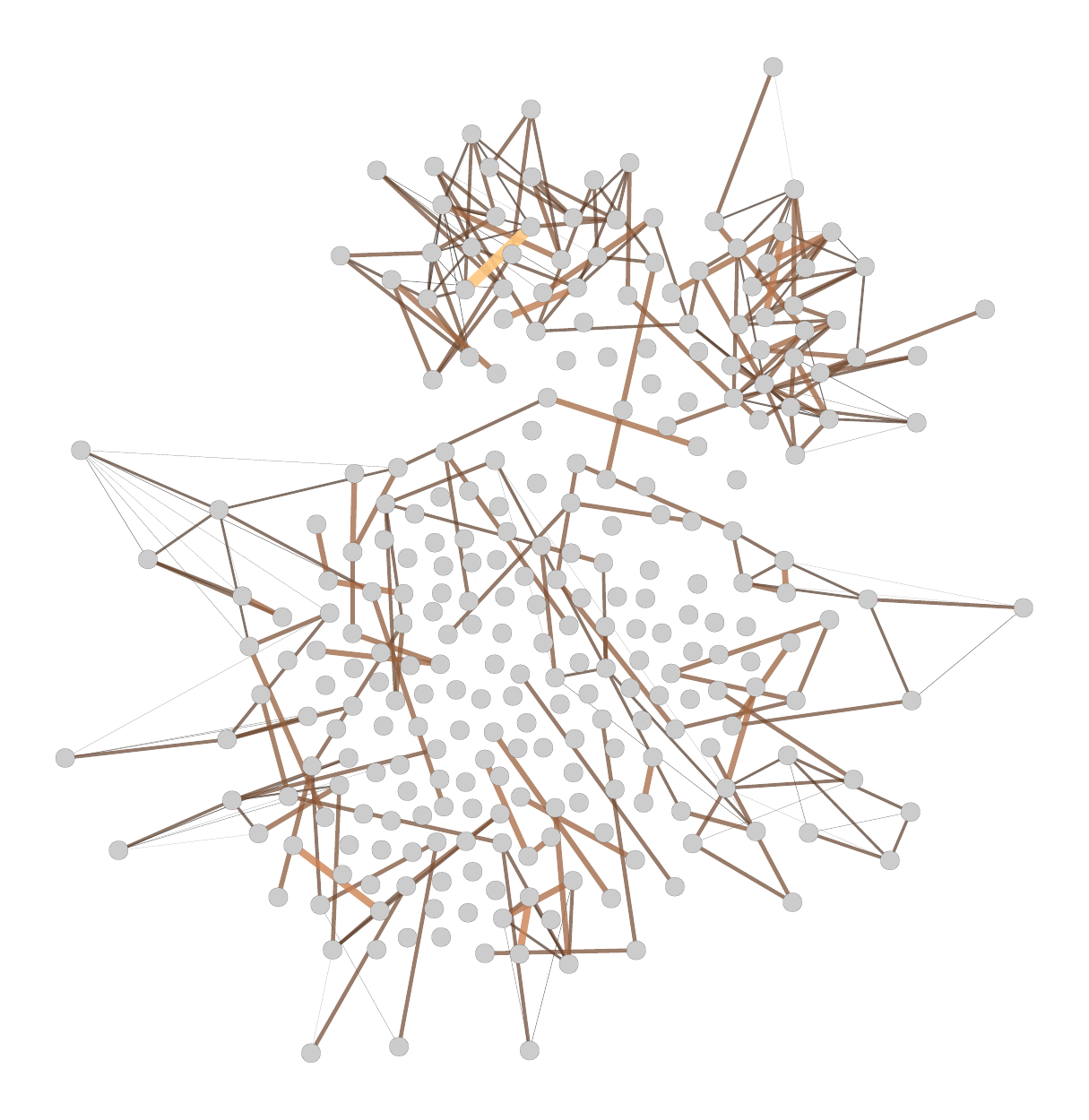} &
    \includegraphics[width=.25\columnwidth]{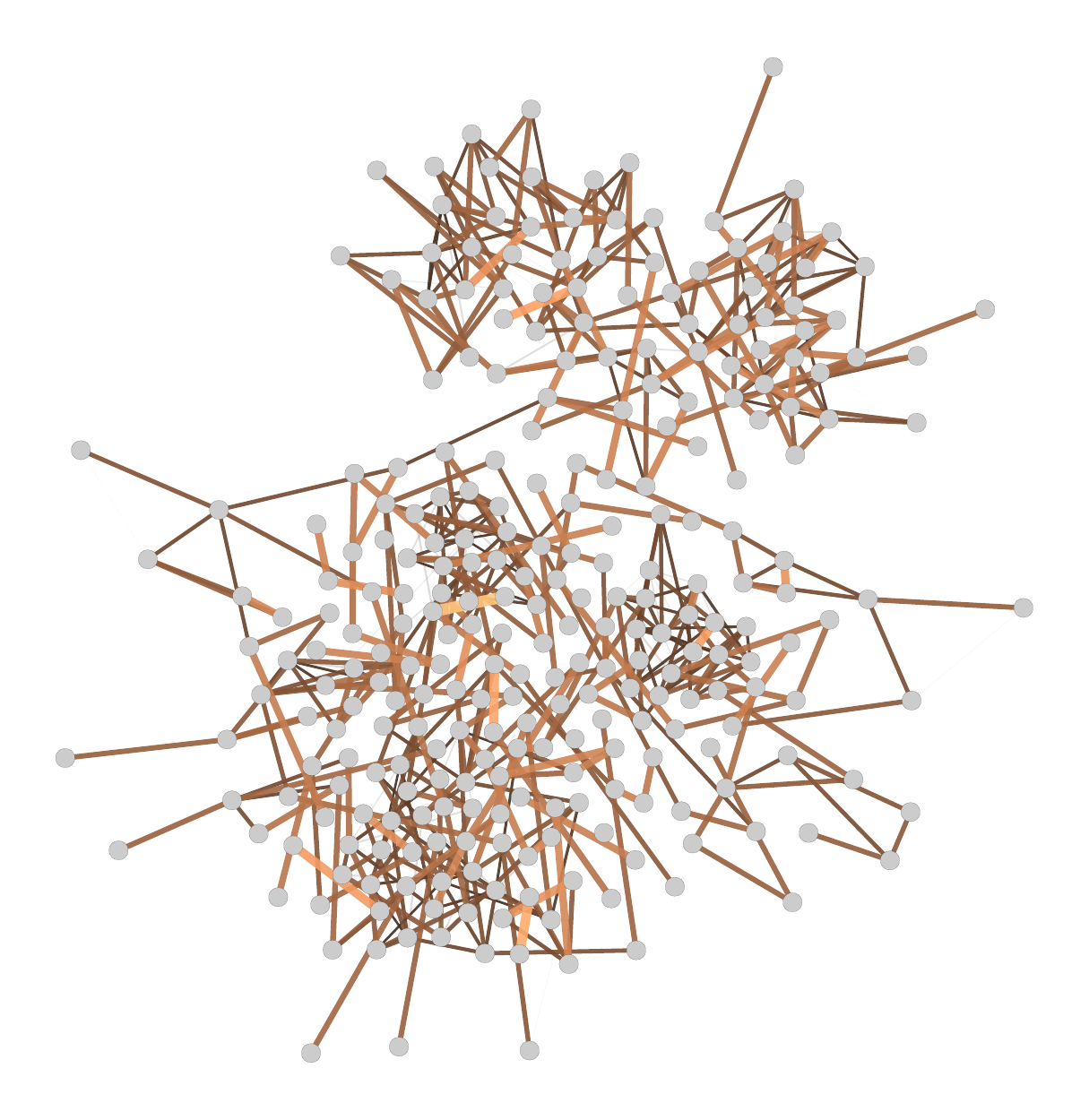} &
    \includegraphics[width=.25\columnwidth]{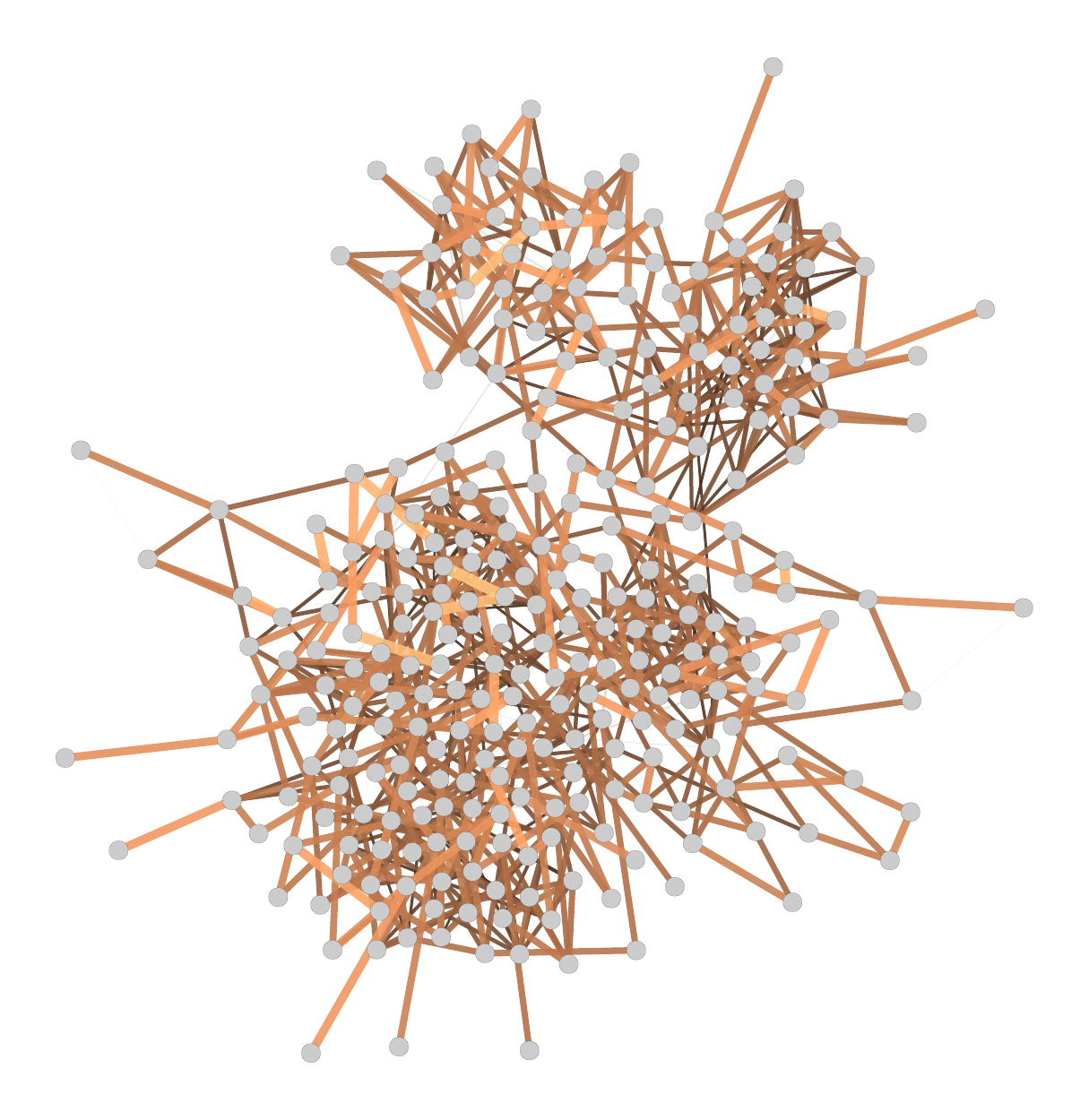} &
    \includegraphics[width=.25\columnwidth]{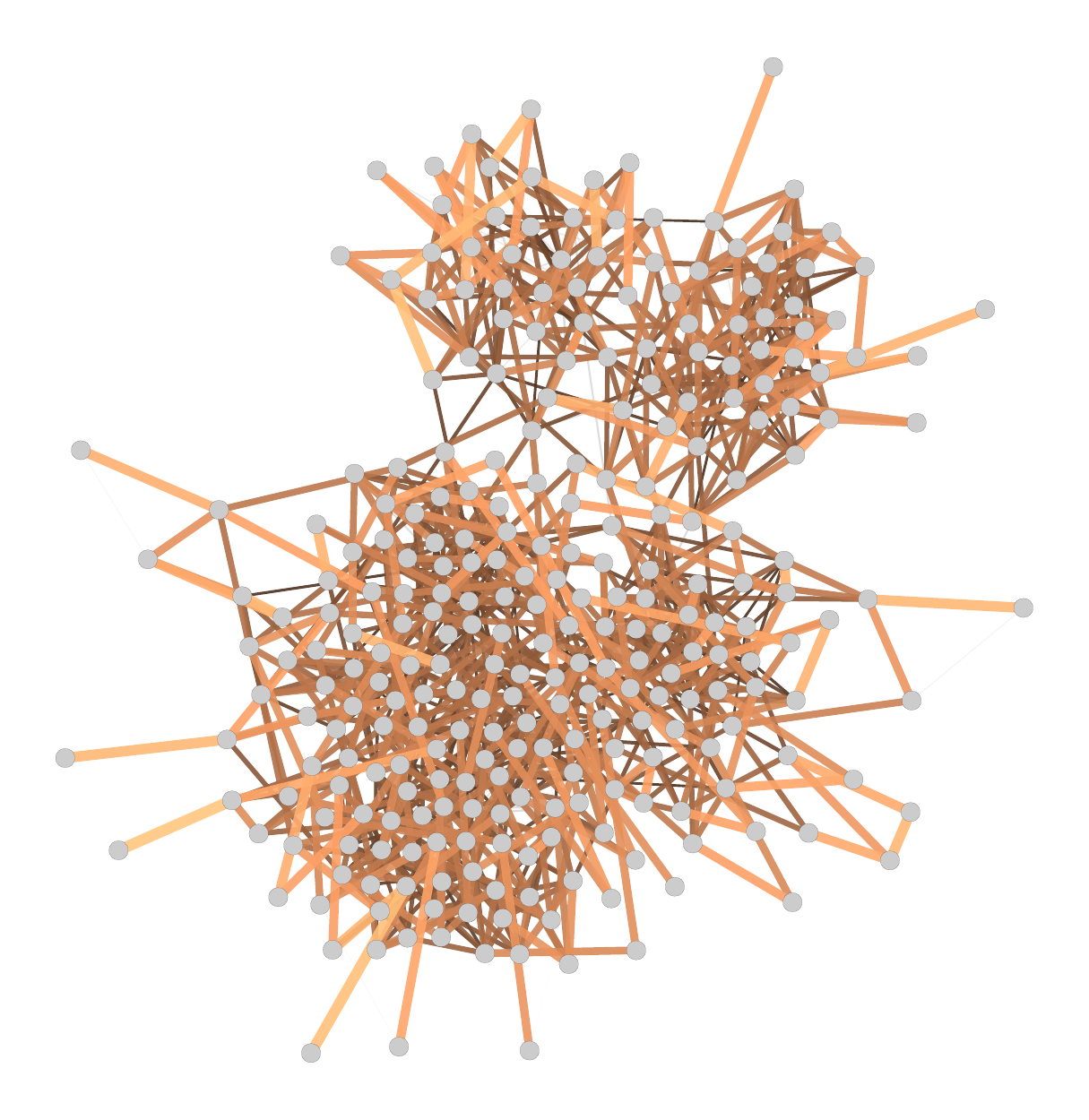}\\
    Iteration 1 & Iteration 2 & Iteration 3 & Iteration 4
  \end{tabular}}

\caption{Example of our greedy coordinate descent algorithm for a covariance
  selection problem on a simulated dataset composed of $M=500$ samples given a
  network of friendships among high-school students~\cite{moody_race_2001}. The
  panels show intermediary results of the algorithm, starting from an empty
  network [i.e. $W_{ij} = 0$ for all $(i,j)$]. The top row shows (a) the random
  initialization of NNDescent (algorithm~\ref{alg:knn}) with $k=4$, (b) its
  final result, and (c) the exact result found with an exhaustive algorithm. The
  middle row shows (d) the result of the $m=\kappa N$ best updates using
  algorithm~\ref{alg:mclosest} with $\kappa=1$ and (e) the exact result
  according to an exhaustive algorithm. The edge colors indicate the value of
  $\max_{W_{ij}}\pi(\W)$. The bottom row shows the first four iterations of the
  GCD algorithm. \label{fig:example}}
\end{figure}

Steps (1) and (2) in the algorithm can be both performed in time $O(1)$ by
keeping a (e.g. Fibonacci) heap containing the nearest $k$ neighbors for each
node. Thus, each full iteration of algorithm~\ref{alg:knn} runs in time
$O(k^{2}N)$, if the degrees of nodes in $\bm U$ are all $O(k)$, otherwise it
would run in time $O(\avg{q}N)$ where $\avg{q}$ is the average number of second
neighbors in $\bm U$. In case $\avg{q} \gg k^2$, the inner loop is then modified
to run over only the first $k$ neighbors for each node, to preserve the
$O(k^{2}N)$ complexity, regardless of any structure that $\bm U$ may have.
Although this algorithm has seen wide deployment, in particular as part of the
popular UMAP dimensionality reduction method~\cite{mcinnes_umap_2018}, and has
had its performance empirically ``battle tested'' in a variety of practical
workloads, it has so far resisted a formal analysis of its algorithmic
complexity. To date, the most careful analyses of this algorithm observe that
the number of iterations required for convergence does not exceed
$2\lceil\log_{2k} N\rceil$ in empirical settings~\cite{baron_empirical_2022}:
The intuitive reasoning is that the initial random graph has as a diameter of
approximately $\lceil\log_{2k} N\rceil$, and hence twice this is the number of
steps needed for each node to communicate its neighborhood to all other nodes,
and for the updated information to return. This conjectured bound on the
convergence results in an overall $O(k^{2}N\log N)$ complexity. However, this
conjecture can only be rigorously proven on a version of the algorithm where the
second neighbor search is replaced by a range query, and the data is generated
by a homogeneous Poisson process~\cite{baron_k-nearest_2020}. These conditions
are not applicable to our problem (and we not assume them), however as we will
see in our numerical experiments, we find robust evidence that it holds to the
case of network reconstruction as well. This typical log-linear complexity of
the NNDescent algorithm is what finally enables us to escape the quadratic
complexity of the CD algorithm, as we will demonstrate shortly.

Importantly, the NNDescent algorithm is approximative, as it does not guarantee
that all nearest neighbors are always correctly identified. Although in practice
it often yields very good recall rates~\cite{dong_efficient_2011}, its inexact
nature is not a concern for our purposes, since it does not affect the
correctness of our GCD algorithm: if an element of the best set
$\mathcal E_{\text{best}}$ in algorithm~\ref{alg:gcd} is missed at a given
iteration, it will eventually be considered in a future iteration, due to the
random initialization of algorithm~\ref{alg:knn}. Our primary concern is only
with the average speed with which it finds the best update candidates.
Nevertheless, as we show later, inaccuracies in the solution of the KNN problem
often disappear completely by the time algorithm~\ref{alg:mclosest} finishes (as
the KNN problem considers a much larger set of pairs than what gets eventually
selected), such that the recall rates for the $m$ closest pairs problem approach
the optimal values as the parameter $\kappa$ is increased.

In Fig~\ref{fig:example} we show an example of the intermediate results obtained
by our algorithm on simulated data on an empirical network.

In appendix~\ref{app:optimizations} we list some low level optimizations that
can improve the overall performance of the algorithm, including parallelization.
A C++ implementation of the algorithm (with OpenMP CPU parallelism) is available
as part of the \texttt{graph-tool} Python
library~\cite{peixoto_graph-tool_2014}.

\subsection{Algorithmic complexity}

We now obtain the overall algorithmic complexity of our GCD algorithm. Since it
consists in repeatedly finding the best $\kappa N$ entries in the matrix $\W$ to
be updated, its overall runtime will depend crucially on the
$\text{\textsc{FindBest}}(\kappa N, \mathcal{V}, \text{\textsc{d}})$ function
defined in algorithm~\ref{alg:mclosest}.
\begin{lemma}
  Assuming the GCD algorithm~\ref{alg:gcd} converges after $\tau$ iterations,
  with $\tau$ independent of $N$---implying that there are at most $O(N)$
  nonzero entries in $\bm W$---then its overall algorithmic complexity is
  determined solely by the
  $\text{\textsc{FindBest}}(\kappa N, \mathcal{V}, \text{\textsc{d}})$ function.
\end{lemma}
\begin{proof}
  At each iteration, algorithm~\ref{alg:gcd} spends time $T(N,m)$ on
  algorithm~\ref{alg:mclosest} to find the $m=\kappa N$ best update coordinates,
  and time $O(\kappa N)$ to actually update them. Therefore, the algorithmic
  complexity per iteration will be given by $T(N,\kappa N)$, since this is
  bounded from below by $\kappa N$. If the slowest execution of algorithm~\ref{alg:mclosest} takes
  time $T(N, \kappa N)$, then the overall algorithm has complexity
  $O(\tau T(N,\kappa N)) = O(T(N, \kappa N))$.
\end{proof}

\begin{lemma}\label{lemma:findbest}
  Assuming the function
  $\text{\textsc{FindKNN}}(k, \mathcal{V}, \text{\textsc{d}})$ defined in
  algorithm~\ref{alg:knn}, which implements
  NNDescent~\cite{dong_efficient_2011}, completes in time $O(k^{2}N\log N)$ (as
  conjectured in Refs.~\cite{baron_k-nearest_2020,baron_empirical_2022}), for
  $N=|\mathcal {V}|$, then the function
  $\text{\textsc{FindBest}}(\kappa N, \mathcal{V}, \textsc{d})$ completes in
  time
  \begin{equation}
    O\left[\left(\sum_{t=0}^{r}\frac{1}{s_{t}}\right)\kappa^{2}N\log N\right],
  \end{equation}
  with $s_{t}=N_{t}/N$, where $N_{t}$ is the number of nodes being considered at
  recursion level $t$, and $t=r+1$ is the first recursion level that results in
  $N_{t}\leq 2\sqrt{\kappa N}$.
\end{lemma}
\begin{proof}
  Our proof follows closely Ref.~\cite{lenhof_k_1992}, in which
  $\text{\textsc{FindKNN}}$ was required to finish in time $O(kN\log N)$,
  instead of $O(k^{2}N\log N)$ as is expected for NNDescent.

When using NNDescent, algorithm~\ref{alg:mclosest} has a complexity given
recursively by
\begin{equation}
  T(N,m) = O(k^{2}N\log N + m\log m) + T(|\mathcal S'|, m),
\end{equation}
with $k=\ceil{4m/N}$, and boundary condition $T(N,m) = O(N^{2})$ if
$N^{2} \leq 4m$ (the $m\log m$ term accounts for the sorting of the $m$ pairs just before the function returns).

In general, ignoring an overall multiplicative constant, we can write
\begin{equation}\label{eq:T}
  T(N,m) = \sum_{t=0}^{r}\frac{m^{2}}{N_{t}}\log N_{t} + (r+1)m\log m + O(m)\\
\end{equation}
where $N_{t}=|\mathcal S_{t-1}'|$ is the number of nodes being considered at
recursion $t$, with $\mathcal S_{t}'$ being the set $\mathcal S'$ at
recursion $t$ (assuming $\mathcal S_{-1}'=\mathcal V$), and $t=r+1$ is the first
point at which $N_{t}\leq 2\sqrt{m}$, and hence the final recursion runs in time
$O(m)$. Introducing $s_{t}=N_{t}/N$ as the fraction of nodes at recursion $t$,
we can write
\begin{multline}
  T(N,\kappa N) = \kappa^{2}N\left[(\log N)\sum_{t=0}^{r}\frac{1}{s_{t}} + \sum_{t=0}^{r}\frac{\log s_{t}}{s_{t}}\right]\\
  + (r+1)\kappa N\log\kappa N + O(\kappa N).
\end{multline}
Since $\log s_{t} \leq 0$ and $1/s_{t}\ge 1$, this will lead to an overall
complexity of
\begin{equation}\label{eq:comp}
  T(N,\kappa N) = O\left[\left(\sum_{t=0}^{r}\frac{1}{s_{t}}\right)\kappa^{2}N\log N\right].
\end{equation}
\end{proof}
The prefactor in the above expression will in general depend on the data and the
stage of the algorithm, as we will see.

We can now obtain a loose upper bound to the running time of \textsc{FindBest}
by assuming the worst-case where the progression of the algorithm is (in
principle) the slowest.

\begin{theorem}[Upper bound on \textsc{FindBest}]\label{theorem:findbest}
  Under the same conditions as in lemma~\ref{lemma:findbest}, the function $\text{\textsc{FindBest}}(\kappa N, \mathcal{V},
  \textsc{d})$ completes in time upper bounded by
  \begin{equation}
    O(\kappa^{3/2} N^{3/2}\log N).
  \end{equation}
\end{theorem}
\begin{proof}
We first
observe that we must always have $s_{t} \leq 1/2^{t}$, i.e. the number of nodes
being considered must decay at least exponentially fast with each recursion.
This is because in algorithm~\ref{alg:mclosest} we have that
$|\mathcal D^{+}_{t}|\ge k |\mathcal S'_{t}|$ and $|\mathcal D^{+}_{t}| = 2m$,
and thus
$N_{t+1} =|\mathcal S'_{t}| \leq 2m/k = 2m/\ceil{4m/N_{t}} \leq N_{t}/2$, and
hence $s_{t+1}\leq s_{t}/2$, which leads to $s_{t} \leq 1/2^{t}$ since
$s_{0}=1$. Therefore, the worst case is $s_{t}=1/2^{t}$, giving us
$\sum_{t=0}^{r}1/s_{t}=\sum_{t=0}^{r}2^{t}=2^{r+1}-1$, and
$2^{r+1} = \sqrt{N/4\kappa}$, and hence a complexity of
\begin{equation}
  T(N,\kappa N) = O(\kappa^{3/2} N^{3/2}\log N).
\end{equation}
\end{proof}

However, although already subquadratic, this upper bound is not tight. This is
because it is not in fact possible for the worst case $s_{t}=1/2^{t}$ to be
realized at every recursion, and in fact the number of nodes being considered
will generically decrease faster than exponentially. We notice this by
performing a more detailed analysis of the runtime, as follows.

\begin{lemma}
  The worst case $s_{t}=1/2^{t}$ considered in the proof of theorem~\ref{theorem:findbest} is not realizable.
\end{lemma}
\begin{proof}
  Let $\bm A$ be the graph consisting of $N$ nodes and the $m$ closest pairs we
  want to find as edges (i.e. the output of the \textsc{FindBest} function).
  Further, let $P(d)$ be the degree distribution of $\bm A$ (i.e. the fraction
  of nodes with degree $d$), and $F(d)=\sum_{d'=d}^{\infty}P(d')$ the tail
  cumulative distribution function of $P(d)$. At recursion $t$ we discover at
  most $k_{t}=\ceil{4m/N_{t}}=\ceil{4\kappa/s_{t}}$ neighbors of each node in
  $\bm A$. Therefore, every node in $\A$ with degree $d\geq k_{t}$ will belong
  to set $\mathcal S'_{t}$ and be considered for next recursion $t+1$. This
  means we can write
\begin{equation}
  N_{t+1}=NF(k_{t}),
\end{equation}
and dividing by $N$ on both sides leads to
\begin{equation}
  s_{t+1}=F(\ceil{4\kappa/s_{t}}),\label{eq:s}
\end{equation}
which no longer depends on $N$.

At this point we can see why the worst case $s_{t}=1/2^{t}$ considered for the
upper bound in theorem~\ref{theorem:findbest} is strictly impossible: it would
correspond to $F(d)=2\kappa/d$ for $d\geq 4\kappa$ (this can be verified by
inserting $F(d)$ in Eq.~\ref{eq:s} and iterating) which is incompatible with the
mean of $P(d)$ being finite and equal to $2\kappa$ --- a requirement of the
graph $\A$ having $N$ nodes and $m=\kappa N$ edges. We notice this by writing
the mean as $2\kappa = \sum_{d}dP(d)= \sum_{d}d[F(d) - F(d+1)]=\sum_{d}F(d)-1$,
and the fact that the sum $\sum_{d=4\kappa}^{\infty}1/d$ diverges.
\end{proof}

Given some valid $F(d)$ we can obtain the prefactor in Eq.~\ref{eq:comp} using
the fact that the algorithm terminates when $s_{r+1} \leq \sqrt{4\kappa/N}$, and
therefore we can recursively invert Eq.~\ref{eq:s} as
\begin{align}
  \frac{1}{s_{r}}&=\frac{F^{-1}(\sqrt{4\kappa/N})}{4\kappa} \\
  \frac{1}{s_{t}}&=\frac{F^{-1}(s_{t+1})}{4\kappa}, \qquad (\text{for }0 < t < r).\label{eq:sr}
\end{align}
The recursion depth and fraction of nodes considered will depend on the degree
distribution of $\bm A$ via $F^{-1}(z)$, and therefore it will vary for different
instances of the problem.

We now consider a few representative cases for how the output of
\textsc{FindBest} may be, and consider their respective running times.

\begin{theorem}
  If the output of $\text{\textsc{FindBest}}(\kappa N, \mathcal{V},
  \textsc{d})$ consists of a $d$-regular graph, then its running
  time is $O(\kappa^{2}N\log^{2}N)$.
\end{theorem}
\begin{proof}
If the graph $\bm A$ is $d$-regular, i.e. every node has degree exactly
$2\kappa$ (assuming this is an integer), corresponding to the extreme case of
maximum homogeneity, this gives us $F(d)=\{1 \text{ if
} d \le 2\kappa,\; 0 \text{ otherwise}\}$, and therefore it follows directly
from Eq.~\ref{eq:s} that $s_{t}= 0$ for $t > 0$, and hence the overall
complexity becomes simply
\begin{equation}
  T(N,\kappa N) = O(\kappa^{2}N\log N),
\end{equation}
corresponding to a single call of the KNN algorithm.
\end{proof}

We move now to a case with intermediary heterogeneity, with the output following
a geometric degreee distribution.
\begin{theorem}
  If the output of $\text{\textsc{FindBest}}(\kappa N, \mathcal{V}, \textsc{d})$
  consists of a graph with a geometric degree distribution $P(d) = (1-p)^{d}p$,
  with $(1-p)/p = 2\kappa$, then its running time becomes
  \begin{equation}
    O(\kappa^{2} N\log^{2} N).
  \end{equation}
\end{theorem}
\begin{proof}
  A geometric distribution gives us $F(d)=[2\kappa/(2\kappa+1)]^{d}$, and
  $F^{-1}(z) =\log z/\log [2\kappa/(2\kappa + 1)]$. Inserting this in
  Eq.~\ref{eq:sr} yields
\begin{align}
  \frac{1}{s_r} &= \frac{\log \sqrt{N/(4\kappa)}}{4\kappa\log[(2\kappa +1)/(2\kappa)]} = O(\log N)\\
  \frac{1}{s_{r-1}} &= \frac{\log 1/s_{r}}{4\kappa\log[(2\kappa +1)/(2\kappa)]} = O(\log\log N)\\
  \frac{1}{s_{r-2}} &= \frac{\log 1/s_{r-1}}{4\kappa\log[(2\kappa +1)/(2\kappa)]} = O(\log\log\log N),
\end{align}
and so on, such that the term $1/s_{r}$ dominates, giving us an overall
log-linear complexity
\begin{equation}
  T(N,\kappa N) = O(\kappa^{2} N\log^{2} N).
\end{equation}
\end{proof}
This result means that for relatively more heterogeneous degree distributions we
accrue only an additional $\log N$ factor in comparison to the $d$-regular case,
and remain fairly below the upper bound found previously.

Based on the above, we can expect broader degree distributions in $\A$ to cause
longer run times. A more extreme case is given by the Zipf distribution, which we now consider.

\begin{theorem}
  If the output of $\text{\textsc{FindBest}}(\kappa N, \mathcal{V}, \textsc{d})$
  consists of a graph with Zipf degree distribution $P(d)=d^{-\alpha}/\zeta(\alpha)$, where $\zeta(\alpha)$ is the Riemann zeta
  function, with $\alpha$ chosen so that the mean is $2\kappa$, then its run time becomes
  \begin{equation}
    O(\kappa^{1-\frac{1}{2(\alpha-1)}} N^{1+\frac{1}{2(\alpha-1)}}\log N).
  \end{equation}
\end{theorem}
\begin{proof}
In this case we
can approximate
$F(d)=\sum_{d'=d}^{\infty}P(d)\approx\zeta(\alpha)^{-1}\int_{d}^{\infty}x^{-\alpha}\dd x\propto d^{1-\alpha}$,
and $F^{-1}(z)\propto z^{1/(1-\alpha)}$. Substituting above in Eq.~\ref{eq:sr}
we get
\begin{align}
  \frac{1}{s_r} &= \frac{\left(\sqrt{4\kappa/N}\right)^{\frac{1}{1-\alpha}}}{4\kappa} = O\left(\kappa^{\frac{1}{2(1-\alpha)}-1}N^{\frac{1}{2(\alpha-1)}}\right)\\
  \frac{1}{s_{r-1}} &= \frac{\left(1/s_{r}\right)^{\frac{1}{\alpha-1}}}{4\kappa} = O\left(\kappa^{\frac{1}{2(1-\alpha)}-2}N^{\frac{1}{2(\alpha-1)^{2}}}\right)\\
  \frac{1}{s_{r-\ell}} &= \frac{\left(1/s_{r-\ell+1}\right)^{\frac{1}{\alpha-1}}}{4\kappa} = O\left(\kappa^{\frac{1}{2(1-\alpha)}-\ell-1}N^{\frac{1}{2(\alpha-1)^{\ell + 1}}}\right),
\end{align}
which is again dominated by $1/s_{r}$, and hence gives us
\begin{equation}
  T(N,\kappa N) = O(\kappa^{1-\frac{1}{2(\alpha-1)}} N^{1+\frac{1}{2(\alpha-1)}}\log N).
\end{equation}
\end{proof}
The value of $\alpha$ is not a free parameter, since it needs to be compatible
with the mean degree $2\kappa$. For very large $2\kappa\gg 1$ we have that
$\alpha\to 2$, and hence the complexity will be asymptotically similar to the
upper bound we found previously, i.e.
\begin{equation}
  T(N,\kappa N) = O(\kappa^{1/2} N^{3/2}\log N),
\end{equation}
although this is not how the algorithm is realistically evoked, as we need
$\kappa = O(1)$ for a reasonable performance. For example, if $\alpha=5/2$ we
have $2\kappa \approx 1.947$, and hence a complexity of
$O(\kappa^{2/3} N^{4/3}\log N)$, and for low $2\kappa \to 1$ we have
$\alpha\to\infty$, yielding the lower limit
\begin{equation}
  T(N,N/2) = O(N\log N),
\end{equation}
compatible with the $d$-regular case for $d=1$, as expected. Therefore, even in
such an extremely heterogeneous case, the complexity remains close to
log-linear for reasonably small values of $\kappa$.

We emphasize that the graph $\bm A$ considered for the complexity analysis above
is distinct from the final network $\bm W$ we want to reconstruct. The latter
might have an arbitrary structure, but the graph $\bm A$ represents only the
updates that need to be performed, and it has a density which controlled by the
parameter $\kappa$ of our algorithm. Thus, even if $\W$ has a very broad degree
distribution, the one we see in $\bm A$ will be further limited by the parameter
$\kappa$ and which updates are needed by the GCD algorithm [for an example,
compare panels (d) and (f) in Fig.~\ref{fig:example}].

\section{Performance assessment}\label{sec:performance}

\begin{figure}
  \includegraphics[width=\columnwidth]{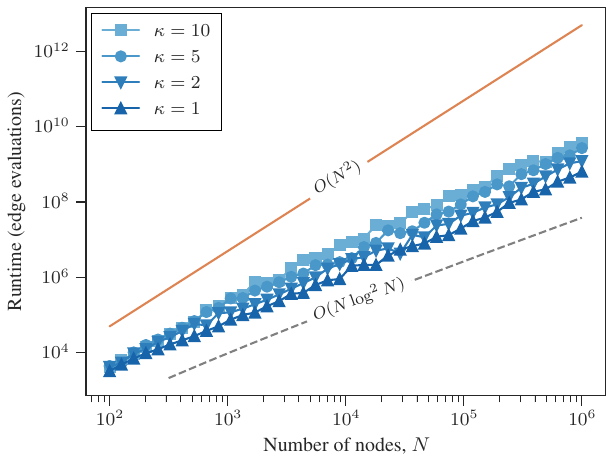}
  \caption{Runtime of the \textsc{FindBest} function
    (algorithm~\ref{alg:mclosest}), for different values of $\kappa$, on $M=10$
    samples of a multivariate Gaussian (see Appendix~\ref{app:models}) on $N$
    nodes and nonzero entries of $\bm W$ sampled as an Erd\H{o}s-Rényi graph
    with mean degree $5$ and nonzero weights independently normally sampled with
    mean $-10^{3}$ and standard deviation 10, and diagonal entries
    $W_{ii}=\sum_{j\ne i}|W_{ij}|/(1-\epsilon)^{2}$ with $\epsilon=10^{-3}$. The
    results show averages over $10$ independent problem
    instances. \label{fig:scaling}}
\end{figure}

\begin{figure}
  \begin{tabular}{cc}
    (a) Brazilian congress votes~\cite{peixoto_network_2019}& (b) Ocean microbiome~\cite{sunagawa_structure_2015}\\
    ($N=882$, Ising model) & ($N=35,651$, Ising model)\\
    \includegraphics[width=.47\columnwidth]{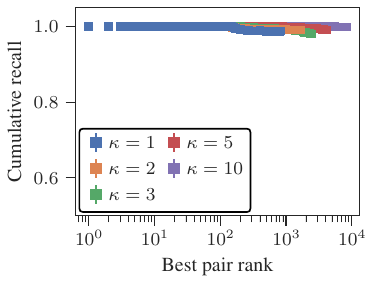} &
    \includegraphics[width=.47\columnwidth]{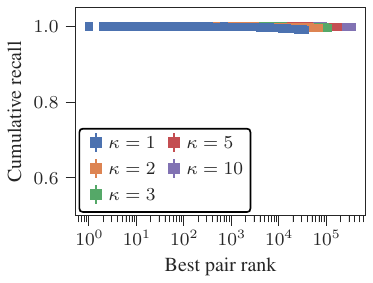}\\
    (c) Rat gene co-expression~\cite{obayashi_coxpresdb_2023}& (d) Microbiome (elbow joint)~\cite{huttenhower_structure_2012}\\
    ($N=13,751$, Multiv. Gaussian) & ($N=10,242$, Ising model)\\
    \includegraphics[width=.47\columnwidth]{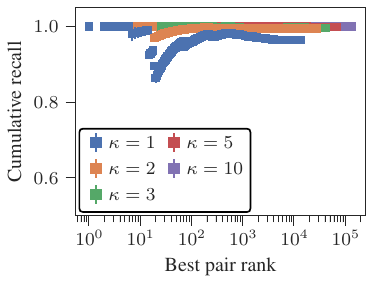} &
    \includegraphics[width=.47\columnwidth]{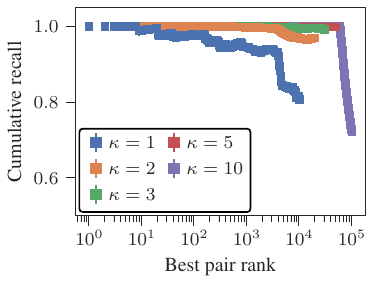}
  \end{tabular}
  \caption{Cumulative recall rates for the \textsc{FindBest} function for
    different values of $\kappa$ on a variety of empirical data and
    reconstruction objectives, as shown in the legend (see
    Appendix~\ref{app:models}). The results shows averages over 10 runs of the
    algorithm.\label{fig:recall}}
\end{figure}

\begin{figure}
  \includegraphics[width=\columnwidth]{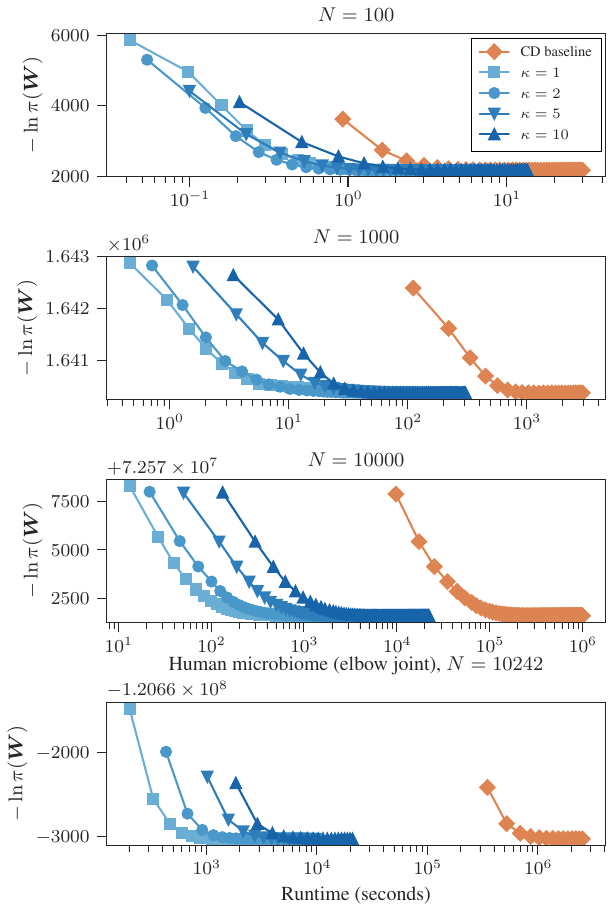}
  \caption{Convergence of the CD and GCD algorithms for artificial data sampled
    from a multivariate Gaussian (same parameterization as Fig.~\ref{fig:scaling}
    but with $M=100$ samples) for three different values of the number of nodes
    $N$ and values of $\kappa$, together with the CD baseline. The bottom panel
    shows the results obtained for empirical data for the human microbiome
    samples of the elbow joint, using the Ising model. \label{fig:conv}}
\end{figure}

\begin{figure}\centering
  \includegraphics[width=.9\columnwidth]{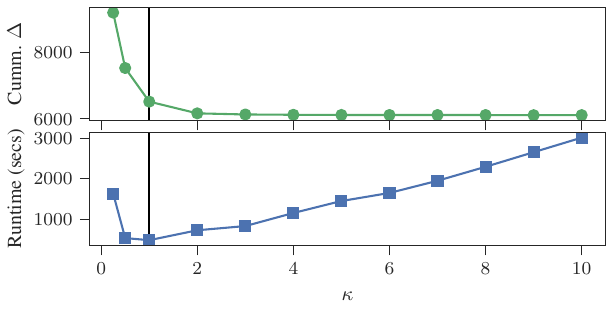}
  \caption{Speed of convergence of the GCD algorithm as a function of $\kappa$
    for the same artificial data as in Fig.~\ref{fig:conv}, for $N=10^{4}$. The
    top panel shows the cumulative sum of the absolute values of $\Delta$ in
    Algorithm~\ref{alg:gcd}, and the bottom panel the total runtime in seconds
    until convergence. The solid lines corresponds to $\kappa =1$. \label{fig:conv2}}
\end{figure}

\begin{figure*}
  \begin{tabular}{c@{\hskip -2\tabcolsep}c@{\hskip -2\tabcolsep}c}
    (a) Earth microbiome & (b) Human microbiome & (c) Nematode gene expression \\
    ($N=317,314$, $M=23,828$) & ($N=45,383$, $M=4,788$) & ($N=17,256$, $M=1,357$)\\
    \begin{tabular}{c}{\includegraphics[width=.32\textwidth,trim=.4cm .4cm .4cm .4cm,clip]{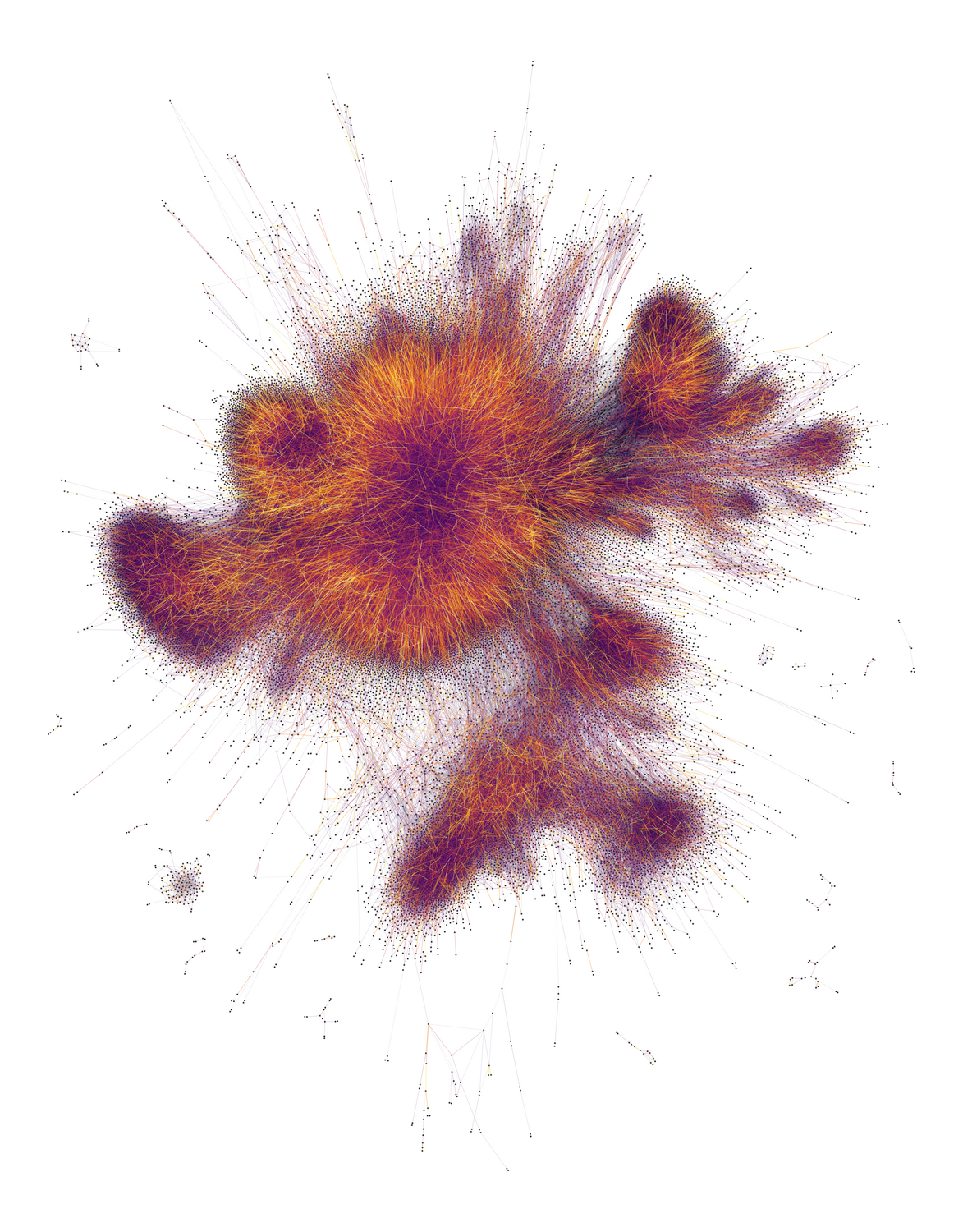}}\end{tabular} &
    \begin{tabular}{c}{\includegraphics[width=.32\textwidth]{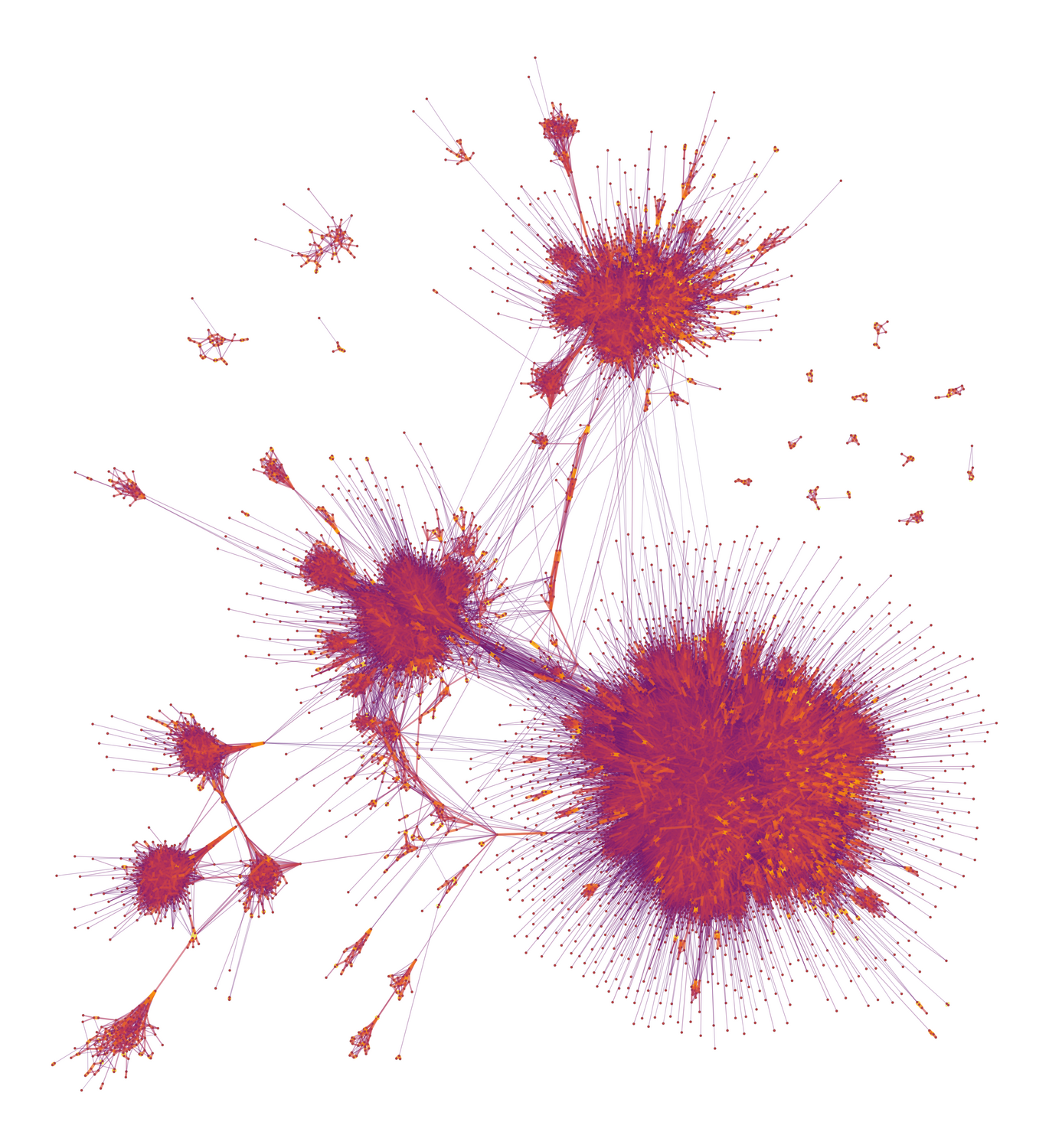}}\end{tabular} &
    \begin{tabular}{c}{\includegraphics[width=.32\textwidth]{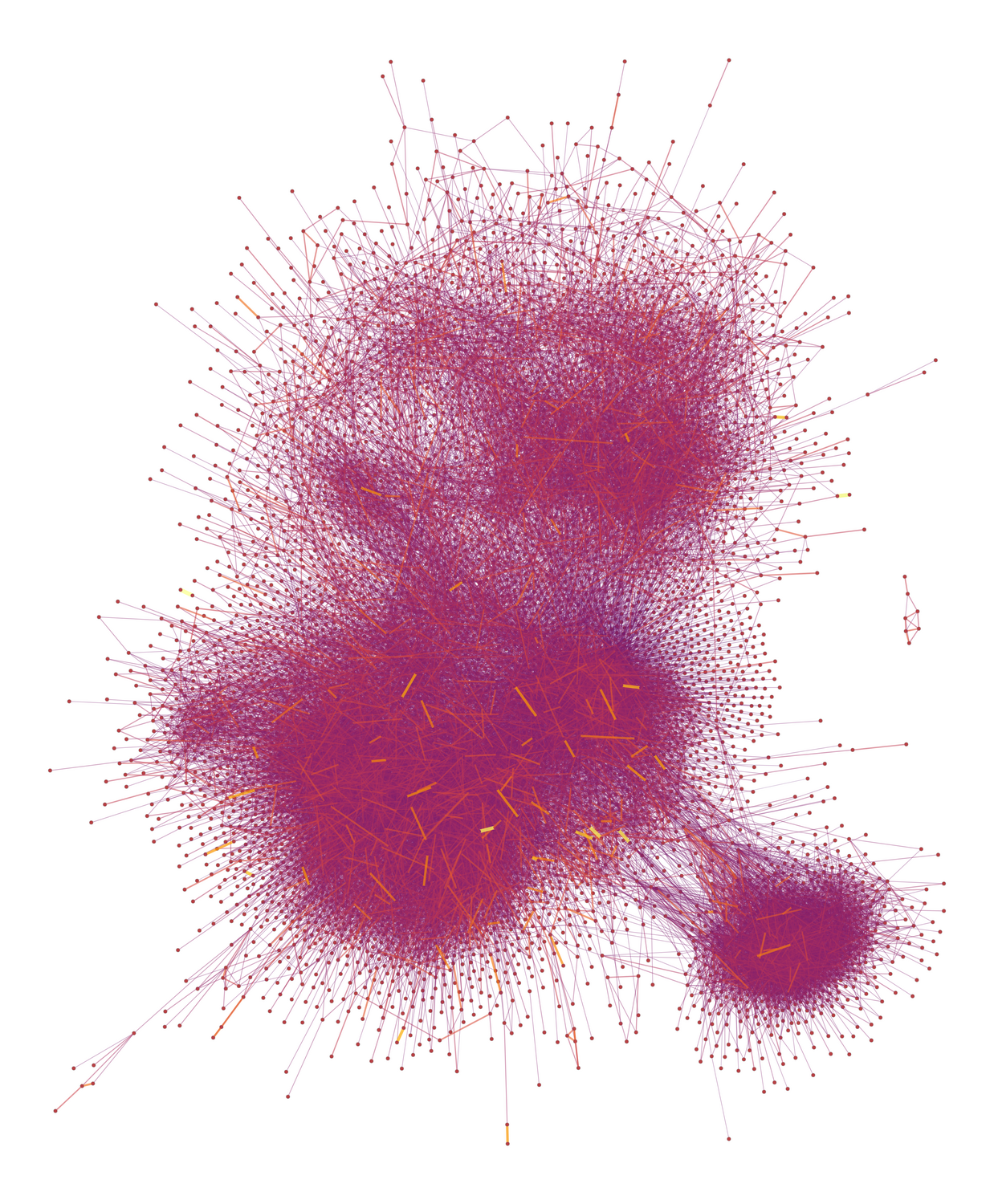}}\end{tabular}
  \end{tabular}
  \caption{Reconstructed networks for the empirical datasets described in the
    text, as shown in the legend, using the Ising model for (a) and (b), and
    multivariate Gaussian for (c). The edge colors indicate the magnitude of the
    entries of the $\bm W$ matrix. Nodes with degree zero were excluded from the
    visualization. \label{fig:empirical}}
\end{figure*}

We now conduct an analysis of the performance of our algorithm in a variety of
artificial and empirical settings. We begin with an analysis of the runtime of
the \textsc{FindBest} function (algorithm~\ref{alg:mclosest}) on $M=10$
samples of a multivariate Gaussian (see Appendix~\ref{app:models}) on $N$ nodes
and nonzero entries of $\bm W$ sampled as an Erd\H{o}s-Rényi graph with mean
degree $5$. As we see in Fig.~\ref{fig:scaling}, the runtime of the algorithm is
consistent with a $O(N\log^{2}N)$ scaling as obtained theoretically in the
previous section for a homogeneous update graph.

This kind of log-linear performance is promising for the scalability of the
overall algorithm, but it stills needs to be determined if the results of
\textsc{FindBest} are sufficiently accurate for a speedy progression of the GCD
algorithm. In Fig.~\ref{fig:recall} we show the cumulative recall rates of the
exact best pairs obtained with an exhaustive algorithm, defined as the fraction
of best pairs correctly identified up to a given rank position (with the best
pair having rank 1), for a variety of empirical data. We observe in general very
good recall rates, compatible with what is obtained with the NNDescent
algorithm~\cite{dong_efficient_2011}. Importantly, in every case we considered,
we observed that the cumulative recall values start at $1$, meaning that the
first best edges are always correctly identified. For some data, the algorithm
may falter slightly for intermediary pairs and a low $\kappa$, but increasing
$\kappa$ has the systematic effect of substantially improving the recall rates
(at the expense of longer runtimes). We emphasize again that it is not strictly
necessary for the \textsc{FindBest} function to return exact results, since it
will be called multiple times during the GCD loop, and its random initialization
guarantees that every pair will eventually be considered --- it needs only to be
able to locate the best edge candidates with high probability. Combined with the
previous result of Fig.~\ref{fig:scaling}, this indicates that the fast runtime
of the \textsc{FindBest} function succeeds in providing the GCD algorithm with
the entries of the $\bm W$ matrix that are most likely to improve the objective
function, while avoiding an exhaustive $O(N^{2})$ search.

We can finally evaluate the final performance of the GCD algorithm by its
convergence speed, as shown in Fig.~\ref{fig:conv} for artificial and empirical
data. For a small data with $N=100$ we observe only a modest improvement over
the CD baseline, but this quickly improves for larger $N$: For $N=1,000$ we
already observe a $100\times$ runtime improvement, which increases to
$1,000\times$ for $N=10,000$. Interestingly, we observe that the $\kappa=1$
results show the fastest convergence, indicating that the decreased accuracy of
the \textsc{FindBest} function --- resulting in it needing to be called more
often --- is compensated by its improved runtime. For $\kappa=10$ we see that
the speed of convergence per iteration is virtually the same as the CD baseline,
but it significantly outperforms it in real time. This seems to demonstrate
that, as expected for the reconstruction of a sparse network, most $O(N^{2})$
updates performed by the CD baseline algorithm are wasteful, and only a $O(N)$
subset is actually needed at each stage for the actual progression of the
algorithm --- which the \textsc{FindBest} function is capable of identifying in
subquadratic time. We explore this in more detail in Fig.~\ref{fig:conv2}, where
we see that for $\kappa>2$ the performance per iteration is virtually
indistinguishable from the CD baseline, but the runtime increases monotonically
with $\kappa > 1$ and decreases with $\kappa < 1$, giving us an optimum at
$\kappa = 1$, which we therefore judge to be the best initial value to be used
in most applications.

The results above are representative of what we have observed on a larger set of
empirical data (not shown). We were not able to find a single instance of an
empirical or artificial scenario that contradicts the log-linear runtime of the
\textsc{FindBest} function, or where our algorithm does not provide a
significant improvement over the $O(N^{2})$ CD baseline (except for very small
$N$).

\section{Empirical examples}\label{sec:empirical}
We demonstrate the use of our algorithm with a few large scale datasets, which
would be costly to analyze with a quadratic reconstruction algorithm. We showcase
on the following:

\paragraph{Earth microbiome project~\cite{thompson_communal_2017}.} A collection
of crowd-sourced microbial samples from various biomes and habitats across the
globe, containing abundances of operational taxonomic units (OTU), obtained via
DNA sequencing and mass spectrometry. An OTU is a proxy for a microbial species,
and a single sample consists of the abundances of each OTU measured at the same
time. This dataset consists of $M=23,828$ samples involving $N=317,314$ OTUs.

\paragraph{Human microbiome
  project~\cite{huttenhower_structure_2012}.} A collection
of microbial samples from 300 healthy adults between the ages of 18 and 40,
collected at five major body sites (oral cavity, nasal cavity, skin,
gastrointestinal tract and urogenital tract) with a total of 15 or 18 specific
body sites. The abundances of OTUs were obtained using 16S rRNA and whole
metagenome shotgun (mWGS) sequencing. This dataset consists of $M=4,788$ samples
involving $N=45,383$ OTUs.

For both co-occurrence datasets above, we binarized each sample as
$x_{i}\in \{-1,1\}$, for absence and presence, respectively, and used the Ising
model for the network reconstruction (see Appendix~\ref{app:models}). In this
case, the matrix $\bm W$ yields the strength of the coupling between two OTUs,
and a value $W_{ij}=0$ means that OTUs $i$ and $j$ are conditionally independent
from one another.

\paragraph{Animal gene expression database COXPRESdb~\cite{obayashi_coxpresdb_2023}.}
This database consists of genome-wide gene expression reads for 10 animal
species as well as budding and fission yeast. Each gene expression sample was
obtained using RNAseq, with read counts converted to base-2 logarithms after
adding a pseudo-count of $0.125$, and batch corrected using
Combat~\cite{johnson_adjusting_2007}. We used the nematode dataset, consisting
of $M=1,357$ samples of $N=17,256$ genes. For this dataset we used the
multivariate Gaussian model (see Appendix.~\ref{app:models}), where the matrix
$\bm W$ corresponds to the inverse covariance between gene expression levels
(a.k.a. precision matrix). In this case, the partial correlation between two
genes, i.e. their degree of association controlling for all other genes, is
given by $-W_{ij}/\sqrt{W_{ii}W_{jj}}$, so that gene pairs with $W_{ij}=0$ are
conditionally independent.

The reconstructed networks for all three datasets are shown in
Fig.~\ref{fig:empirical}. They all seem to display prominent modular structure.
In the case of microbial co-occurrence datasets the clusters correspond mostly to
different habitats and geographical regions for the earth microbiome, and to
different body sites for the human microbiome.

\section{Conclusion}\label{sec:conclusion}

We have described a method to reconstruct interaction networks from
observational data that avoids a seemingly inherent quadratic complexity in the
number of nodes, in favor of a data-dependent runtime that is typically
log-linear. Our algorithm does not rely on particular formulations of the
reconstruction problem, other than the updates on the edge weights being done
constant time with respect to the total number of nodes. Together with its
straightforward parallelizability, our proposed method removes a central barrier
to the reconstruction of large-scale networks, and can be applied to problems
with a number of nodes and edges on the order of hundreds of thousands,
millions, or potentially even more depending on available computing resources.

Our algorithm relies on the NNDescent~\cite{dong_efficient_2011} approach for
approximate $k$-nearest neighbor search. Despite the robust empirical evidence
for its performance, a detailed theoretical analysis of this algorithm is still
lacking, in particular of its potential modes of failures. We expect further
progress in this direction to elucidate potential limitations and improvements
to our overall approach.

In this work we focused on convex reconstruction objectives, such as the inverse
Ising model and multivariate Gaussian with $L_{1}$ regularization. More robust
regularization schemes or different models may no longer be convex, in which
case coordinate descent will fail in general at converging to the global
optimum. However, it is clear that our strategy of finding the best edge
candidates in subquadratic time is also applicable for algorithms that can be
used with nonconvex objectives, such as stochastic gradient descent or simulated
annealing. We leave such extensions for future work.

\section{Acknowledgments}
Sample processing, sequencing, and core amplicon data analysis were performed by
the Earth Microbiome Project (\url{http://www.earthmicrobiome.org}), and all amplicon
sequence data and metadata have been made public through the EMP data portal
(\url{http://qiita.microbio.me/emp}).

\bibliography{bib}

\begin{thebibliography}{44}%
\makeatletter
\providecommand \@ifxundefined [1]{%
 \@ifx{#1\undefined}
}%
\providecommand \@ifnum [1]{%
 \ifnum #1\expandafter \@firstoftwo
 \else \expandafter \@secondoftwo
 \fi
}%
\providecommand \@ifx [1]{%
 \ifx #1\expandafter \@firstoftwo
 \else \expandafter \@secondoftwo
 \fi
}%
\providecommand \natexlab [1]{#1}%
\providecommand \enquote  [1]{``#1''}%
\providecommand \bibnamefont  [1]{#1}%
\providecommand \bibfnamefont [1]{#1}%
\providecommand \citenamefont [1]{#1}%
\providecommand \href@noop [0]{\@secondoftwo}%
\providecommand \href [0]{\begingroup \@sanitize@url \@href}%
\providecommand \@href[1]{\@@startlink{#1}\@@href}%
\providecommand \@@href[1]{\endgroup#1\@@endlink}%
\providecommand \@sanitize@url [0]{\catcode `\\12\catcode `\$12\catcode
  `\&12\catcode `\#12\catcode `\^12\catcode `\_12\catcode `\%12\relax}%
\providecommand \@@startlink[1]{}%
\providecommand \@@endlink[0]{}%
\providecommand \url  [0]{\begingroup\@sanitize@url \@url }%
\providecommand \@url [1]{\endgroup\@href {#1}{\urlprefix }}%
\providecommand \urlprefix  [0]{URL }%
\providecommand \Eprint [0]{\href }%
\providecommand \doibase [0]{https://doi.org/}%
\providecommand \selectlanguage [0]{\@gobble}%
\providecommand \bibinfo  [0]{\@secondoftwo}%
\providecommand \bibfield  [0]{\@secondoftwo}%
\providecommand \translation [1]{[#1]}%
\providecommand \BibitemOpen [0]{}%
\providecommand \bibitemStop [0]{}%
\providecommand \bibitemNoStop [0]{.\EOS\space}%
\providecommand \EOS [0]{\spacefactor3000\relax}%
\providecommand \BibitemShut  [1]{\csname bibitem#1\endcsname}%
\let\auto@bib@innerbib\@empty
\bibitem [{\citenamefont {Dong}\ \emph {et~al.}(2011)\citenamefont {Dong},
  \citenamefont {Moses},\ and\ \citenamefont {Li}}]{dong_efficient_2011}%
  \BibitemOpen
  \bibfield  {author} {\bibinfo {author} {\bibfnamefont {W.}~\bibnamefont
  {Dong}}, \bibinfo {author} {\bibfnamefont {C.}~\bibnamefont {Moses}},\ and\
  \bibinfo {author} {\bibfnamefont {K.}~\bibnamefont {Li}},\ }\bibfield
  {title} {\bibinfo {title} {Efficient k-nearest neighbor graph construction
  for generic similarity measures},\ }in\ \href
  {https://doi.org/10.1145/1963405.1963487} {\emph {\bibinfo {booktitle}
  {Proceedings of the 20th International Conference on {{World}} Wide Web}}},\
  \bibinfo {series and number} {{{WWW}} '11}\ (\bibinfo  {publisher}
  {Association for Computing Machinery},\ \bibinfo {address} {New York, NY,
  USA},\ \bibinfo {year} {2011})\ pp.\ \bibinfo {pages} {577--586}\BibitemShut
  {NoStop}%
\bibitem [{\citenamefont {Baron}\ and\ \citenamefont
  {Darling}(2020)}]{baron_k-nearest_2020}%
  \BibitemOpen
  \bibfield  {author} {\bibinfo {author} {\bibfnamefont {J.~D.}\ \bibnamefont
  {Baron}}\ and\ \bibinfo {author} {\bibfnamefont {R.~W.~R.}\ \bibnamefont
  {Darling}},\ }\href {https://doi.org/10.48550/arXiv.1908.07645} {\bibinfo
  {title} {K-{{Nearest Neighbor Approximation Via}} the {{Friend-of-a-Friend
  Principle}}}} (\bibinfo {year} {2020}),\ \Eprint
  {https://arxiv.org/abs/1908.07645} {arxiv:1908.07645 [math, stat]}
  \BibitemShut {NoStop}%
\bibitem [{\citenamefont {Baron}\ and\ \citenamefont
  {Darling}(2022)}]{baron_empirical_2022}%
  \BibitemOpen
  \bibfield  {author} {\bibinfo {author} {\bibfnamefont {J.~D.}\ \bibnamefont
  {Baron}}\ and\ \bibinfo {author} {\bibfnamefont {R.~W.~R.}\ \bibnamefont
  {Darling}},\ }\href {https://doi.org/10.48550/arXiv.2202.00517} {\bibinfo
  {title} {Empirical complexity of comparator-based nearest neighbor descent}}
  (\bibinfo {year} {2022}),\ \Eprint {https://arxiv.org/abs/2202.00517}
  {arxiv:2202.00517 [cs, stat]} \BibitemShut {NoStop}%
\bibitem [{\citenamefont {Lauritzen}(1996)}]{lauritzen_graphical_1996}%
  \BibitemOpen
  \bibfield  {author} {\bibinfo {author} {\bibfnamefont {S.~L.}\ \bibnamefont
  {Lauritzen}},\ }\href
  {https://books.google.com/books?hl=de&lr=&id=mGQWkx4guhAC&oi=fnd&pg=PA1&dq=S.+L.+Lauritzen.+Graphical+Models.+Clarendon+Press,+Oxford,+UK,+1996.&ots=2Mcjq8Ct1e&sig=z-vlgEMAXMZwc1K0bcrpFOUzIGg}
  {\emph {\bibinfo {title} {Graphical Models}}},\ Vol.~\bibinfo {volume} {17}\
  (\bibinfo  {publisher} {Clarendon Press},\ \bibinfo {year}
  {1996})\BibitemShut {NoStop}%
\bibitem [{\citenamefont {Jordan}(2004)}]{jordan_graphical_2004}%
  \BibitemOpen
  \bibfield  {author} {\bibinfo {author} {\bibfnamefont {M.~I.}\ \bibnamefont
  {Jordan}},\ }\bibfield  {title} {\bibinfo {title} {Graphical {{Models}}},\
  }\href {https://doi.org/10.1214/088342304000000026} {\bibfield  {journal}
  {\bibinfo  {journal} {Statistical Science}\ }\textbf {\bibinfo {volume}
  {19}},\ \bibinfo {pages} {140} (\bibinfo {year} {2004})}\BibitemShut
  {NoStop}%
\bibitem [{\citenamefont {Drton}\ and\ \citenamefont
  {Maathuis}(2017)}]{drton_structure_2017}%
  \BibitemOpen
  \bibfield  {author} {\bibinfo {author} {\bibfnamefont {M.}~\bibnamefont
  {Drton}}\ and\ \bibinfo {author} {\bibfnamefont {M.~H.}\ \bibnamefont
  {Maathuis}},\ }\bibfield  {title} {\bibinfo {title} {Structure {{Learning}}
  in {{Graphical Modeling}}},\ }\href
  {https://doi.org/10.1146/annurev-statistics-060116-053803} {\bibfield
  {journal} {\bibinfo  {journal} {Annual Review of Statistics and Its
  Application}\ }\textbf {\bibinfo {volume} {4}},\ \bibinfo {pages} {365}
  (\bibinfo {year} {2017})}\BibitemShut {NoStop}%
\bibitem [{\citenamefont {Timme}\ and\ \citenamefont
  {Casadiego}(2014)}]{timme_revealing_2014}%
  \BibitemOpen
  \bibfield  {author} {\bibinfo {author} {\bibfnamefont {M.}~\bibnamefont
  {Timme}}\ and\ \bibinfo {author} {\bibfnamefont {J.}~\bibnamefont
  {Casadiego}},\ }\bibfield  {title} {\bibinfo {title} {Revealing networks from
  dynamics: An introduction},\ }\href
  {https://doi.org/10.1088/1751-8113/47/34/343001} {\bibfield  {journal}
  {\bibinfo  {journal} {Journal of Physics A: Mathematical and Theoretical}\
  }\textbf {\bibinfo {volume} {47}},\ \bibinfo {pages} {343001} (\bibinfo
  {year} {2014})}\BibitemShut {NoStop}%
\bibitem [{\citenamefont {Hallac}\ \emph {et~al.}(2017)\citenamefont {Hallac},
  \citenamefont {Park}, \citenamefont {Boyd},\ and\ \citenamefont
  {Leskovec}}]{hallac_network_2017}%
  \BibitemOpen
  \bibfield  {author} {\bibinfo {author} {\bibfnamefont {D.}~\bibnamefont
  {Hallac}}, \bibinfo {author} {\bibfnamefont {Y.}~\bibnamefont {Park}},
  \bibinfo {author} {\bibfnamefont {S.}~\bibnamefont {Boyd}},\ and\ \bibinfo
  {author} {\bibfnamefont {J.}~\bibnamefont {Leskovec}},\ }\href
  {https://doi.org/10.48550/arXiv.1703.01958} {\bibinfo {title} {Network
  {{Inference}} via the {{Time-Varying Graphical Lasso}}}} (\bibinfo {year}
  {2017}),\ \Eprint {https://arxiv.org/abs/1703.01958} {arxiv:1703.01958 [cs,
  math]} \BibitemShut {NoStop}%
\bibitem [{\citenamefont {Guseva}\ \emph {et~al.}(2022)\citenamefont {Guseva},
  \citenamefont {Darcy}, \citenamefont {Simon}, \citenamefont {Alteio},
  \citenamefont {{Montesinos-Navarro}},\ and\ \citenamefont
  {Kaiser}}]{guseva_diversity_2022}%
  \BibitemOpen
  \bibfield  {author} {\bibinfo {author} {\bibfnamefont {K.}~\bibnamefont
  {Guseva}}, \bibinfo {author} {\bibfnamefont {S.}~\bibnamefont {Darcy}},
  \bibinfo {author} {\bibfnamefont {E.}~\bibnamefont {Simon}}, \bibinfo
  {author} {\bibfnamefont {L.~V.}\ \bibnamefont {Alteio}}, \bibinfo {author}
  {\bibfnamefont {A.}~\bibnamefont {{Montesinos-Navarro}}},\ and\ \bibinfo
  {author} {\bibfnamefont {C.}~\bibnamefont {Kaiser}},\ }\bibfield  {title}
  {\bibinfo {title} {From diversity to complexity: {{Microbial}} networks in
  soils},\ }\href {https://doi.org/10.1016/j.soilbio.2022.108604} {\bibfield
  {journal} {\bibinfo  {journal} {Soil Biology and Biochemistry}\ }\textbf
  {\bibinfo {volume} {169}},\ \bibinfo {pages} {108604} (\bibinfo {year}
  {2022})}\BibitemShut {NoStop}%
\bibitem [{\citenamefont {Bury}(2013)}]{bury_statistical_2013}%
  \BibitemOpen
  \bibfield  {author} {\bibinfo {author} {\bibfnamefont {T.}~\bibnamefont
  {Bury}},\ }\bibfield  {title} {\bibinfo {title} {A statistical physics
  perspective on criticality in financial markets},\ }\href
  {https://doi.org/10.1088/1742-5468/2013/11/P11004} {\bibfield  {journal}
  {\bibinfo  {journal} {Journal of Statistical Mechanics: Theory and
  Experiment}\ }\textbf {\bibinfo {volume} {2013}},\ \bibinfo {pages} {P11004}
  (\bibinfo {year} {2013})},\ \Eprint {https://arxiv.org/abs/1310.2446}
  {arxiv:1310.2446 [physics, q-fin]} \BibitemShut {NoStop}%
\bibitem [{\citenamefont {Weigt}\ \emph {et~al.}(2009)\citenamefont {Weigt},
  \citenamefont {White}, \citenamefont {Szurmant}, \citenamefont {Hoch},\ and\
  \citenamefont {Hwa}}]{weigt_identification_2009}%
  \BibitemOpen
  \bibfield  {author} {\bibinfo {author} {\bibfnamefont {M.}~\bibnamefont
  {Weigt}}, \bibinfo {author} {\bibfnamefont {R.~A.}\ \bibnamefont {White}},
  \bibinfo {author} {\bibfnamefont {H.}~\bibnamefont {Szurmant}}, \bibinfo
  {author} {\bibfnamefont {J.~A.}\ \bibnamefont {Hoch}},\ and\ \bibinfo
  {author} {\bibfnamefont {T.}~\bibnamefont {Hwa}},\ }\bibfield  {title}
  {\bibinfo {title} {Identification of direct residue contacts in
  protein--protein interaction by message passing},\ }\href
  {https://doi.org/10.1073/pnas.0805923106} {\bibfield  {journal} {\bibinfo
  {journal} {Proceedings of the National Academy of Sciences}\ }\textbf
  {\bibinfo {volume} {106}},\ \bibinfo {pages} {67} (\bibinfo {year}
  {2009})}\BibitemShut {NoStop}%
\bibitem [{\citenamefont {D'haeseleer}\ \emph {et~al.}(2000)\citenamefont
  {D'haeseleer}, \citenamefont {Liang},\ and\ \citenamefont
  {Somogyi}}]{dhaeseleer_genetic_2000}%
  \BibitemOpen
  \bibfield  {author} {\bibinfo {author} {\bibfnamefont {P.}~\bibnamefont
  {D'haeseleer}}, \bibinfo {author} {\bibfnamefont {S.}~\bibnamefont {Liang}},\
  and\ \bibinfo {author} {\bibfnamefont {R.}~\bibnamefont {Somogyi}},\
  }\bibfield  {title} {\bibinfo {title} {Genetic network inference: From
  co-expression clustering to reverse engineering},\ }\href
  {https://doi.org/10.1093/bioinformatics/16.8.707} {\bibfield  {journal}
  {\bibinfo  {journal} {Bioinformatics}\ }\textbf {\bibinfo {volume} {16}},\
  \bibinfo {pages} {707} (\bibinfo {year} {2000})}\BibitemShut {NoStop}%
\bibitem [{\citenamefont {Manning}\ \emph {et~al.}(2014)\citenamefont
  {Manning}, \citenamefont {Ranganath}, \citenamefont {Norman},\ and\
  \citenamefont {Blei}}]{manning_topographic_2014}%
  \BibitemOpen
  \bibfield  {author} {\bibinfo {author} {\bibfnamefont {J.~R.}\ \bibnamefont
  {Manning}}, \bibinfo {author} {\bibfnamefont {R.}~\bibnamefont {Ranganath}},
  \bibinfo {author} {\bibfnamefont {K.~A.}\ \bibnamefont {Norman}},\ and\
  \bibinfo {author} {\bibfnamefont {D.~M.}\ \bibnamefont {Blei}},\ }\bibfield
  {title} {\bibinfo {title} {Topographic {{Factor Analysis}}: {{A Bayesian
  Model}} for {{Inferring Brain Networks}} from {{Neural Data}}},\ }\href
  {https://doi.org/10.1371/journal.pone.0094914} {\bibfield  {journal}
  {\bibinfo  {journal} {PLOS ONE}\ }\textbf {\bibinfo {volume} {9}},\ \bibinfo
  {pages} {e94914} (\bibinfo {year} {2014})}\BibitemShut {NoStop}%
\bibitem [{\citenamefont {{Braunstein Alfredo}}\ \emph
  {et~al.}(2019)\citenamefont {{Braunstein Alfredo}}, \citenamefont {{Ingrosso
  Alessandro}},\ and\ \citenamefont {{Muntoni Anna
  Paola}}}]{braunstein_alfredo_network_2019}%
  \BibitemOpen
  \bibfield  {author} {\bibinfo {author} {\bibnamefont {{Braunstein Alfredo}}},
  \bibinfo {author} {\bibnamefont {{Ingrosso Alessandro}}},\ and\ \bibinfo
  {author} {\bibnamefont {{Muntoni Anna Paola}}},\ }\bibfield  {title}
  {\bibinfo {title} {Network reconstruction from infection cascades},\ }\href
  {https://doi.org/10.1098/rsif.2018.0844} {\bibfield  {journal} {\bibinfo
  {journal} {Journal of The Royal Society Interface}\ }\textbf {\bibinfo
  {volume} {16}},\ \bibinfo {pages} {20180844} (\bibinfo {year}
  {2019})}\BibitemShut {NoStop}%
\bibitem [{\citenamefont {Dempster}(1972)}]{dempster_covariance_1972}%
  \BibitemOpen
  \bibfield  {author} {\bibinfo {author} {\bibfnamefont {A.~P.}\ \bibnamefont
  {Dempster}},\ }\bibfield  {title} {\bibinfo {title} {Covariance
  {{Selection}}},\ }\href {https://doi.org/10.2307/2528966} {\bibfield
  {journal} {\bibinfo  {journal} {Biometrics}\ }\textbf {\bibinfo {volume}
  {28}},\ \bibinfo {pages} {157} (\bibinfo {year} {1972})},\ \Eprint
  {https://arxiv.org/abs/2528966} {2528966} \BibitemShut {NoStop}%
\bibitem [{\citenamefont {Friedman}\ \emph {et~al.}(2008)\citenamefont
  {Friedman}, \citenamefont {Hastie},\ and\ \citenamefont
  {Tibshirani}}]{friedman_sparse_2008}%
  \BibitemOpen
  \bibfield  {author} {\bibinfo {author} {\bibfnamefont {J.}~\bibnamefont
  {Friedman}}, \bibinfo {author} {\bibfnamefont {T.}~\bibnamefont {Hastie}},\
  and\ \bibinfo {author} {\bibfnamefont {R.}~\bibnamefont {Tibshirani}},\
  }\bibfield  {title} {\bibinfo {title} {Sparse inverse covariance estimation
  with the graphical lasso},\ }\href
  {https://doi.org/10.1093/biostatistics/kxm045} {\bibfield  {journal}
  {\bibinfo  {journal} {Biostatistics}\ }\textbf {\bibinfo {volume} {9}},\
  \bibinfo {pages} {432} (\bibinfo {year} {2008})}\BibitemShut {NoStop}%
\bibitem [{\citenamefont {Mazumder}\ and\ \citenamefont
  {Hastie}(2012)}]{mazumder_graphical_2012}%
  \BibitemOpen
  \bibfield  {author} {\bibinfo {author} {\bibfnamefont {R.}~\bibnamefont
  {Mazumder}}\ and\ \bibinfo {author} {\bibfnamefont {T.}~\bibnamefont
  {Hastie}},\ }\bibfield  {title} {\bibinfo {title} {The graphical lasso:
  {{New}} insights and alternatives},\ }\href
  {https://doi.org/10.1214/12-EJS740} {\bibfield  {journal} {\bibinfo
  {journal} {Electronic Journal of Statistics}\ }\textbf {\bibinfo {volume}
  {6}},\ \bibinfo {pages} {2125} (\bibinfo {year} {2012})}\BibitemShut
  {NoStop}%
\bibitem [{\citenamefont {Hastie}\ \emph {et~al.}(2015)\citenamefont {Hastie},
  \citenamefont {Tibshirani},\ and\ \citenamefont
  {Wainwright}}]{hastie_statistical_2015}%
  \BibitemOpen
  \bibfield  {author} {\bibinfo {author} {\bibfnamefont {T.}~\bibnamefont
  {Hastie}}, \bibinfo {author} {\bibfnamefont {R.}~\bibnamefont {Tibshirani}},\
  and\ \bibinfo {author} {\bibfnamefont {M.}~\bibnamefont {Wainwright}},\
  }\href@noop {} {\emph {\bibinfo {title} {Statistical {{Learning}} with
  {{Sparsity}}: {{The Lasso}} and {{Generalizations}}}}}\ (\bibinfo
  {publisher} {CRC Press},\ \bibinfo {year} {2015})\BibitemShut {NoStop}%
\bibitem [{\citenamefont {Bresler}\ \emph {et~al.}(2008)\citenamefont
  {Bresler}, \citenamefont {Mossel},\ and\ \citenamefont
  {Sly}}]{bresler_reconstruction_2008}%
  \BibitemOpen
  \bibfield  {author} {\bibinfo {author} {\bibfnamefont {G.}~\bibnamefont
  {Bresler}}, \bibinfo {author} {\bibfnamefont {E.}~\bibnamefont {Mossel}},\
  and\ \bibinfo {author} {\bibfnamefont {A.}~\bibnamefont {Sly}},\ }\bibfield
  {title} {\bibinfo {title} {Reconstruction of {{Markov Random Fields}} from
  {{Samples}}: {{Some Observations}} and {{Algorithms}}},\ }in\ \href
  {https://doi.org/10.1007/978-3-540-85363-3_28} {\emph {\bibinfo {booktitle}
  {Approximation, {{Randomization}} and {{Combinatorial Optimization}}.
  {{Algorithms}} and {{Techniques}}}}},\ \bibinfo {series and number} {Lecture
  {{Notes}} in {{Computer Science}}}\ (\bibinfo  {publisher} {Springer, Berlin,
  Heidelberg},\ \bibinfo {year} {2008})\ pp.\ \bibinfo {pages}
  {343--356}\BibitemShut {NoStop}%
\bibitem [{\citenamefont {Nguyen}\ \emph {et~al.}(2017)\citenamefont {Nguyen},
  \citenamefont {Zecchina},\ and\ \citenamefont {Berg}}]{nguyen_inverse_2017}%
  \BibitemOpen
  \bibfield  {author} {\bibinfo {author} {\bibfnamefont {H.~C.}\ \bibnamefont
  {Nguyen}}, \bibinfo {author} {\bibfnamefont {R.}~\bibnamefont {Zecchina}},\
  and\ \bibinfo {author} {\bibfnamefont {J.}~\bibnamefont {Berg}},\ }\bibfield
  {title} {\bibinfo {title} {Inverse statistical problems: From the inverse
  {{Ising}} problem to data science},\ }\href
  {https://doi.org/10.1080/00018732.2017.1341604} {\bibfield  {journal}
  {\bibinfo  {journal} {Advances in Physics}\ }\textbf {\bibinfo {volume}
  {66}},\ \bibinfo {pages} {197} (\bibinfo {year} {2017})}\BibitemShut
  {NoStop}%
\bibitem [{\citenamefont {Hsieh}\ \emph {et~al.}(2014)\citenamefont {Hsieh},
  \citenamefont {Sustik}, \citenamefont {Dhillon},\ and\ \citenamefont
  {Ravikumar}}]{hsieh_quic_2014}%
  \BibitemOpen
  \bibfield  {author} {\bibinfo {author} {\bibfnamefont {C.-J.}\ \bibnamefont
  {Hsieh}}, \bibinfo {author} {\bibfnamefont {M.~A.}\ \bibnamefont {Sustik}},
  \bibinfo {author} {\bibfnamefont {I.~S.}\ \bibnamefont {Dhillon}},\ and\
  \bibinfo {author} {\bibfnamefont {P.}~\bibnamefont {Ravikumar}},\ }\bibfield
  {title} {\bibinfo {title} {{{QUIC}}: {{Quadratic Approximation}} for {{Sparse
  Inverse Covariance Estimation}}},\ }\href
  {http://jmlr.org/papers/v15/hsieh14a.html} {\bibfield  {journal} {\bibinfo
  {journal} {Journal of Machine Learning Research}\ }\textbf {\bibinfo {volume}
  {15}},\ \bibinfo {pages} {2911} (\bibinfo {year} {2014})}\BibitemShut
  {NoStop}%
\bibitem [{\citenamefont {Hsieh}\ \emph {et~al.}(2013)\citenamefont {Hsieh},
  \citenamefont {Sustik}, \citenamefont {Dhillon}, \citenamefont {Ravikumar},\
  and\ \citenamefont {Poldrack}}]{hsieh_big_2013}%
  \BibitemOpen
  \bibfield  {author} {\bibinfo {author} {\bibfnamefont {C.-J.}\ \bibnamefont
  {Hsieh}}, \bibinfo {author} {\bibfnamefont {M.~A.}\ \bibnamefont {Sustik}},
  \bibinfo {author} {\bibfnamefont {I.~S.}\ \bibnamefont {Dhillon}}, \bibinfo
  {author} {\bibfnamefont {P.~K.}\ \bibnamefont {Ravikumar}},\ and\ \bibinfo
  {author} {\bibfnamefont {R.}~\bibnamefont {Poldrack}},\ }\bibfield  {title}
  {\bibinfo {title} {{{BIG}} \& {{QUIC}}: {{Sparse Inverse Covariance
  Estimation}} for a {{Million Variables}}},\ }in\ \href
  {https://proceedings.neurips.cc/paper/2013/hash/1abb1e1ea5f481b589da52303b091cbb-Abstract.html}
  {\emph {\bibinfo {booktitle} {Advances in {{Neural Information Processing
  Systems}}}}},\ Vol.~\bibinfo {volume} {26}\ (\bibinfo  {publisher} {Curran
  Associates, Inc.},\ \bibinfo {year} {2013})\BibitemShut {NoStop}%
\bibitem [{\citenamefont {Bresler}(2015)}]{bresler_efficiently_2015}%
  \BibitemOpen
  \bibfield  {author} {\bibinfo {author} {\bibfnamefont {G.}~\bibnamefont
  {Bresler}},\ }\bibfield  {title} {\bibinfo {title} {Efficiently {{Learning
  Ising Models}} on {{Arbitrary Graphs}}},\ }in\ \href
  {https://doi.org/10.1145/2746539.2746631} {\emph {\bibinfo {booktitle}
  {Proceedings of the {{Forty-seventh Annual ACM Symposium}} on {{Theory}} of
  {{Computing}}}}},\ \bibinfo {series and number} {{{STOC}} '15}\ (\bibinfo
  {publisher} {ACM},\ \bibinfo {address} {New York, NY, USA},\ \bibinfo {year}
  {2015})\ pp.\ \bibinfo {pages} {771--782}\BibitemShut {NoStop}%
\bibitem [{\citenamefont {Bresler}\ \emph {et~al.}(2010)\citenamefont
  {Bresler}, \citenamefont {Mossel},\ and\ \citenamefont
  {Sly}}]{bresler_reconstruction_2010}%
  \BibitemOpen
  \bibfield  {author} {\bibinfo {author} {\bibfnamefont {G.}~\bibnamefont
  {Bresler}}, \bibinfo {author} {\bibfnamefont {E.}~\bibnamefont {Mossel}},\
  and\ \bibinfo {author} {\bibfnamefont {A.}~\bibnamefont {Sly}},\ }\href
  {https://doi.org/10.48550/arXiv.0712.1402} {\bibinfo {title} {Reconstruction
  of {{Markov Random Fields}} from {{Samples}}: {{Some Easy Observations}} and
  {{Algorithms}}}} (\bibinfo {year} {2010}),\ \Eprint
  {https://arxiv.org/abs/0712.1402} {arxiv:0712.1402 [cs]} \BibitemShut
  {NoStop}%
\bibitem [{\citenamefont {Bresler}\ \emph {et~al.}(2018)\citenamefont
  {Bresler}, \citenamefont {Gamarnik},\ and\ \citenamefont
  {Shah}}]{bresler_learning_2018}%
  \BibitemOpen
  \bibfield  {author} {\bibinfo {author} {\bibfnamefont {G.}~\bibnamefont
  {Bresler}}, \bibinfo {author} {\bibfnamefont {D.}~\bibnamefont {Gamarnik}},\
  and\ \bibinfo {author} {\bibfnamefont {D.}~\bibnamefont {Shah}},\ }\bibfield
  {title} {\bibinfo {title} {Learning {{Graphical Models From}} the {{Glauber
  Dynamics}}},\ }\href {https://doi.org/10.1109/TIT.2017.2713828} {\bibfield
  {journal} {\bibinfo  {journal} {IEEE Transactions on Information Theory}\
  }\textbf {\bibinfo {volume} {64}},\ \bibinfo {pages} {4072} (\bibinfo {year}
  {2018})}\BibitemShut {NoStop}%
\bibitem [{\citenamefont {Tseng}(2001)}]{tseng_convergence_2001}%
  \BibitemOpen
  \bibfield  {author} {\bibinfo {author} {\bibfnamefont {P.}~\bibnamefont
  {Tseng}},\ }\bibfield  {title} {\bibinfo {title} {Convergence of a {{Block
  Coordinate Descent Method}} for {{Nondifferentiable Minimization}}},\ }\href
  {https://doi.org/10.1023/A:1017501703105} {\bibfield  {journal} {\bibinfo
  {journal} {Journal of Optimization Theory and Applications}\ }\textbf
  {\bibinfo {volume} {109}},\ \bibinfo {pages} {475} (\bibinfo {year}
  {2001})}\BibitemShut {NoStop}%
\bibitem [{\citenamefont {Wright}(2015)}]{wright_coordinate_2015}%
  \BibitemOpen
  \bibfield  {author} {\bibinfo {author} {\bibfnamefont {S.~J.}\ \bibnamefont
  {Wright}},\ }\bibfield  {title} {\bibinfo {title} {Coordinate descent
  algorithms},\ }\href {https://doi.org/10.1007/s10107-015-0892-3} {\bibfield
  {journal} {\bibinfo  {journal} {Mathematical Programming}\ }\textbf {\bibinfo
  {volume} {151}},\ \bibinfo {pages} {3} (\bibinfo {year} {2015})}\BibitemShut
  {NoStop}%
\bibitem [{\citenamefont {Spall}(2012)}]{spall_cyclic_2012}%
  \BibitemOpen
  \bibfield  {author} {\bibinfo {author} {\bibfnamefont {J.~C.}\ \bibnamefont
  {Spall}},\ }\bibfield  {title} {\bibinfo {title} {Cyclic {{Seesaw Process}}
  for {{Optimization}} and {{Identification}}},\ }\href
  {https://doi.org/10.1007/s10957-012-0001-1} {\bibfield  {journal} {\bibinfo
  {journal} {Journal of Optimization Theory and Applications}\ }\textbf
  {\bibinfo {volume} {154}},\ \bibinfo {pages} {187} (\bibinfo {year}
  {2012})}\BibitemShut {NoStop}%
\bibitem [{\citenamefont {Abbeel}\ \emph {et~al.}(2006)\citenamefont {Abbeel},
  \citenamefont {Koller},\ and\ \citenamefont {Ng}}]{abbeel_learning_2006}%
  \BibitemOpen
  \bibfield  {author} {\bibinfo {author} {\bibfnamefont {P.}~\bibnamefont
  {Abbeel}}, \bibinfo {author} {\bibfnamefont {D.}~\bibnamefont {Koller}},\
  and\ \bibinfo {author} {\bibfnamefont {A.~Y.}\ \bibnamefont {Ng}},\
  }\bibfield  {title} {\bibinfo {title} {Learning {{Factor Graphs}} in
  {{Polynomial Time}} and {{Sample Complexity}}},\ }\href
  {http://jmlr.org/papers/v7/abbeel06a.html} {\bibfield  {journal} {\bibinfo
  {journal} {Journal of Machine Learning Research}\ }\textbf {\bibinfo {volume}
  {7}},\ \bibinfo {pages} {1743} (\bibinfo {year} {2006})}\BibitemShut
  {NoStop}%
\bibitem [{\citenamefont {Wainwright}\ \emph {et~al.}(2006)\citenamefont
  {Wainwright}, \citenamefont {Lafferty},\ and\ \citenamefont
  {Ravikumar}}]{wainwright_high-dimensional_2006}%
  \BibitemOpen
  \bibfield  {author} {\bibinfo {author} {\bibfnamefont {M.~J.}\ \bibnamefont
  {Wainwright}}, \bibinfo {author} {\bibfnamefont {J.}~\bibnamefont
  {Lafferty}},\ and\ \bibinfo {author} {\bibfnamefont {P.}~\bibnamefont
  {Ravikumar}},\ }\bibfield  {title} {\bibinfo {title} {High-dimensional
  graphical model selection using $\ell_1$-regularized logistic regression},\
  }in\ \href
  {https://proceedings.neurips.cc/paper/2006/hash/86b20716fbd5b253d27cec43127089bc-Abstract.html}
  {\emph {\bibinfo {booktitle} {Advances in Neural Information Processing
  Systems}}},\ Vol.~\bibinfo {volume} {19}\ (\bibinfo  {publisher} {MIT
  Press},\ \bibinfo {year} {2006})\BibitemShut {NoStop}%
\bibitem [{\citenamefont {Santhanam}\ and\ \citenamefont
  {Wainwright}(2009)}]{santhanam_information-theoretic_2009}%
  \BibitemOpen
  \bibfield  {author} {\bibinfo {author} {\bibfnamefont {N.}~\bibnamefont
  {Santhanam}}\ and\ \bibinfo {author} {\bibfnamefont {M.~J.}\ \bibnamefont
  {Wainwright}},\ }\href {https://doi.org/10.48550/arXiv.0905.2639} {\bibinfo
  {title} {Information-theoretic limits of selecting binary graphical models in
  high dimensions}} (\bibinfo {year} {2009}),\ \Eprint
  {https://arxiv.org/abs/0905.2639} {arxiv:0905.2639 [cs, math, stat]}
  \BibitemShut {NoStop}%
\bibitem [{\citenamefont {Smid}(2000)}]{smid_closest-point_2000}%
  \BibitemOpen
  \bibfield  {author} {\bibinfo {author} {\bibfnamefont {M.}~\bibnamefont
  {Smid}},\ }\bibfield  {title} {\bibinfo {title} {Closest-{{Point Problems}}
  in {{Computational Geometry}}},\ }in\ \href
  {https://doi.org/10.1016/B978-044482537-7/50021-8} {\emph {\bibinfo
  {booktitle} {Handbook of {{Computational Geometry}}}}},\ \bibinfo {editor}
  {edited by\ \bibinfo {editor} {\bibfnamefont {J.~R.}\ \bibnamefont {Sack}}\
  and\ \bibinfo {editor} {\bibfnamefont {J.}~\bibnamefont {Urrutia}}}\
  (\bibinfo  {publisher} {North-Holland},\ \bibinfo {address} {Amsterdam},\
  \bibinfo {year} {2000})\ pp.\ \bibinfo {pages} {877--935}\BibitemShut
  {NoStop}%
\bibitem [{\citenamefont {Lenhof}\ and\ \citenamefont
  {Smid}(1992)}]{lenhof_k_1992}%
  \BibitemOpen
  \bibfield  {author} {\bibinfo {author} {\bibfnamefont {H.-P.}\ \bibnamefont
  {Lenhof}}\ and\ \bibinfo {author} {\bibfnamefont {M.}~\bibnamefont {Smid}},\
  }\bibfield  {title} {\bibinfo {title} {The k closest pairs problem},\ }\href
  {https://people.scs.carleton.ca/~michiel/k-closestnote.pdf} {\bibfield
  {journal} {\bibinfo  {journal} {Unpublished manuscript}\ } (\bibinfo {year}
  {1992})}\BibitemShut {NoStop}%
\bibitem [{\citenamefont {Moody}(2001)}]{moody_race_2001}%
  \BibitemOpen
  \bibfield  {author} {\bibinfo {author} {\bibfnamefont {J.}~\bibnamefont
  {Moody}},\ }\bibfield  {title} {\bibinfo {title} {Race, {{School
  Integration}}, and {{Friendship Segregation}} in {{America}}},\ }\href
  {https://doi.org/10.1086/338954} {\bibfield  {journal} {\bibinfo  {journal}
  {American Journal of Sociology}\ }\textbf {\bibinfo {volume} {107}},\
  \bibinfo {pages} {679} (\bibinfo {year} {2001})}\BibitemShut {NoStop}%
\bibitem [{\citenamefont {McInnes}\ \emph {et~al.}(2018)\citenamefont
  {McInnes}, \citenamefont {Healy},\ and\ \citenamefont
  {Melville}}]{mcinnes_umap_2018}%
  \BibitemOpen
  \bibfield  {author} {\bibinfo {author} {\bibfnamefont {L.}~\bibnamefont
  {McInnes}}, \bibinfo {author} {\bibfnamefont {J.}~\bibnamefont {Healy}},\
  and\ \bibinfo {author} {\bibfnamefont {J.}~\bibnamefont {Melville}},\
  }\bibfield  {title} {\bibinfo {title} {{{UMAP}}: {{Uniform Manifold
  Approximation}} and {{Projection}} for {{Dimension Reduction}}},\ }\href
  {http://arxiv.org/abs/1802.03426} {\bibfield  {journal} {\bibinfo  {journal}
  {arXiv:1802.03426 [cs, stat]}\ } (\bibinfo {year} {2018})},\ \Eprint
  {https://arxiv.org/abs/1802.03426} {arxiv:1802.03426 [cs, stat]} \BibitemShut
  {NoStop}%
\bibitem [{\citenamefont {Peixoto}(2014)}]{peixoto_graph-tool_2014}%
  \BibitemOpen
  \bibfield  {author} {\bibinfo {author} {\bibfnamefont {T.~P.}\ \bibnamefont
  {Peixoto}},\ }\bibfield  {title} {\bibinfo {title} {The graph-tool python
  library},\ }\bibfield  {journal} {\bibinfo  {journal} {figshare}\ }\href
  {https://doi.org/10.6084/m9.figshare.1164194} {10.6084/m9.figshare.1164194}
  (\bibinfo {year} {2014}),\ \bibinfo {note} {available at
  \url{https://graph-tool.skewed.de}.}\BibitemShut {Stop}%
\bibitem [{\citenamefont {Peixoto}(2019)}]{peixoto_network_2019}%
  \BibitemOpen
  \bibfield  {author} {\bibinfo {author} {\bibfnamefont {T.~P.}\ \bibnamefont
  {Peixoto}},\ }\bibfield  {title} {\bibinfo {title} {Network
  {{Reconstruction}} and {{Community Detection}} from {{Dynamics}}},\ }\href
  {https://doi.org/10.1103/PhysRevLett.123.128301} {\bibfield  {journal}
  {\bibinfo  {journal} {Physical Review Letters}\ }\textbf {\bibinfo {volume}
  {123}},\ \bibinfo {pages} {128301} (\bibinfo {year} {2019})}\BibitemShut
  {NoStop}%
\bibitem [{\citenamefont {Sunagawa}\ \emph {et~al.}(2015)\citenamefont
  {Sunagawa}, \citenamefont {Coelho}, \citenamefont {Chaffron}, \citenamefont
  {Kultima}, \citenamefont {Labadie}, \citenamefont {Salazar}, \citenamefont
  {Djahanschiri}, \citenamefont {Zeller}, \citenamefont {Mende}, \citenamefont
  {Alberti}, \citenamefont {{Cornejo-Castillo}}, \citenamefont {Costea},
  \citenamefont {Cruaud}, \citenamefont {{d'Ovidio}}, \citenamefont {Engelen},
  \citenamefont {Ferrera}, \citenamefont {Gasol}, \citenamefont {Guidi},
  \citenamefont {Hildebrand}, \citenamefont {Kokoszka}, \citenamefont
  {Lepoivre}, \citenamefont {{Lima-Mendez}}, \citenamefont {Poulain},
  \citenamefont {Poulos}, \citenamefont {{Royo-Llonch}}, \citenamefont
  {Sarmento}, \citenamefont {{Vieira-Silva}}, \citenamefont {Dimier},
  \citenamefont {Picheral}, \citenamefont {Searson}, \citenamefont
  {{Kandels-Lewis}}, \citenamefont {{Tara Oceans coordinators}}, \citenamefont
  {Bowler}, \citenamefont {{de Vargas}}, \citenamefont {Gorsky}, \citenamefont
  {Grimsley}, \citenamefont {Hingamp}, \citenamefont {Iudicone}, \citenamefont
  {Jaillon}, \citenamefont {Not}, \citenamefont {Ogata}, \citenamefont
  {Pesant}, \citenamefont {Speich}, \citenamefont {Stemmann}, \citenamefont
  {Sullivan}, \citenamefont {Weissenbach}, \citenamefont {Wincker},
  \citenamefont {Karsenti}, \citenamefont {Raes}, \citenamefont {Acinas},\ and\
  \citenamefont {Bork}}]{sunagawa_structure_2015}%
  \BibitemOpen
  \bibfield  {author} {\bibinfo {author} {\bibfnamefont {S.}~\bibnamefont
  {Sunagawa}}, \bibinfo {author} {\bibfnamefont {L.~P.}\ \bibnamefont
  {Coelho}}, \bibinfo {author} {\bibfnamefont {S.}~\bibnamefont {Chaffron}},
  \bibinfo {author} {\bibfnamefont {J.~R.}\ \bibnamefont {Kultima}}, \bibinfo
  {author} {\bibfnamefont {K.}~\bibnamefont {Labadie}}, \bibinfo {author}
  {\bibfnamefont {G.}~\bibnamefont {Salazar}}, \bibinfo {author} {\bibfnamefont
  {B.}~\bibnamefont {Djahanschiri}}, \bibinfo {author} {\bibfnamefont
  {G.}~\bibnamefont {Zeller}}, \bibinfo {author} {\bibfnamefont {D.~R.}\
  \bibnamefont {Mende}}, \bibinfo {author} {\bibfnamefont {A.}~\bibnamefont
  {Alberti}}, \bibinfo {author} {\bibfnamefont {F.~M.}\ \bibnamefont
  {{Cornejo-Castillo}}}, \bibinfo {author} {\bibfnamefont {P.~I.}\ \bibnamefont
  {Costea}}, \bibinfo {author} {\bibfnamefont {C.}~\bibnamefont {Cruaud}},
  \bibinfo {author} {\bibfnamefont {F.}~\bibnamefont {{d'Ovidio}}}, \bibinfo
  {author} {\bibfnamefont {S.}~\bibnamefont {Engelen}}, \bibinfo {author}
  {\bibfnamefont {I.}~\bibnamefont {Ferrera}}, \bibinfo {author} {\bibfnamefont
  {J.~M.}\ \bibnamefont {Gasol}}, \bibinfo {author} {\bibfnamefont
  {L.}~\bibnamefont {Guidi}}, \bibinfo {author} {\bibfnamefont
  {F.}~\bibnamefont {Hildebrand}}, \bibinfo {author} {\bibfnamefont
  {F.}~\bibnamefont {Kokoszka}}, \bibinfo {author} {\bibfnamefont
  {C.}~\bibnamefont {Lepoivre}}, \bibinfo {author} {\bibfnamefont
  {G.}~\bibnamefont {{Lima-Mendez}}}, \bibinfo {author} {\bibfnamefont
  {J.}~\bibnamefont {Poulain}}, \bibinfo {author} {\bibfnamefont {B.~T.}\
  \bibnamefont {Poulos}}, \bibinfo {author} {\bibfnamefont {M.}~\bibnamefont
  {{Royo-Llonch}}}, \bibinfo {author} {\bibfnamefont {H.}~\bibnamefont
  {Sarmento}}, \bibinfo {author} {\bibfnamefont {S.}~\bibnamefont
  {{Vieira-Silva}}}, \bibinfo {author} {\bibfnamefont {C.}~\bibnamefont
  {Dimier}}, \bibinfo {author} {\bibfnamefont {M.}~\bibnamefont {Picheral}},
  \bibinfo {author} {\bibfnamefont {S.}~\bibnamefont {Searson}}, \bibinfo
  {author} {\bibfnamefont {S.}~\bibnamefont {{Kandels-Lewis}}}, \bibinfo
  {author} {\bibnamefont {{Tara Oceans coordinators}}}, \bibinfo {author}
  {\bibfnamefont {C.}~\bibnamefont {Bowler}}, \bibinfo {author} {\bibfnamefont
  {C.}~\bibnamefont {{de Vargas}}}, \bibinfo {author} {\bibfnamefont
  {G.}~\bibnamefont {Gorsky}}, \bibinfo {author} {\bibfnamefont
  {N.}~\bibnamefont {Grimsley}}, \bibinfo {author} {\bibfnamefont
  {P.}~\bibnamefont {Hingamp}}, \bibinfo {author} {\bibfnamefont
  {D.}~\bibnamefont {Iudicone}}, \bibinfo {author} {\bibfnamefont
  {O.}~\bibnamefont {Jaillon}}, \bibinfo {author} {\bibfnamefont
  {F.}~\bibnamefont {Not}}, \bibinfo {author} {\bibfnamefont {H.}~\bibnamefont
  {Ogata}}, \bibinfo {author} {\bibfnamefont {S.}~\bibnamefont {Pesant}},
  \bibinfo {author} {\bibfnamefont {S.}~\bibnamefont {Speich}}, \bibinfo
  {author} {\bibfnamefont {L.}~\bibnamefont {Stemmann}}, \bibinfo {author}
  {\bibfnamefont {M.~B.}\ \bibnamefont {Sullivan}}, \bibinfo {author}
  {\bibfnamefont {J.}~\bibnamefont {Weissenbach}}, \bibinfo {author}
  {\bibfnamefont {P.}~\bibnamefont {Wincker}}, \bibinfo {author} {\bibfnamefont
  {E.}~\bibnamefont {Karsenti}}, \bibinfo {author} {\bibfnamefont
  {J.}~\bibnamefont {Raes}}, \bibinfo {author} {\bibfnamefont {S.~G.}\
  \bibnamefont {Acinas}},\ and\ \bibinfo {author} {\bibfnamefont
  {P.}~\bibnamefont {Bork}},\ }\bibfield  {title} {\bibinfo {title} {Structure
  and function of the global ocean microbiome},\ }\href
  {https://doi.org/10.1126/science.1261359} {\bibfield  {journal} {\bibinfo
  {journal} {Science}\ }\textbf {\bibinfo {volume} {348}},\ \bibinfo {pages}
  {1261359} (\bibinfo {year} {2015})}\BibitemShut {NoStop}%
\bibitem [{\citenamefont {Obayashi}\ \emph {et~al.}(2023)\citenamefont
  {Obayashi}, \citenamefont {Kodate}, \citenamefont {Hibara}, \citenamefont
  {Kagaya},\ and\ \citenamefont {Kinoshita}}]{obayashi_coxpresdb_2023}%
  \BibitemOpen
  \bibfield  {author} {\bibinfo {author} {\bibfnamefont {T.}~\bibnamefont
  {Obayashi}}, \bibinfo {author} {\bibfnamefont {S.}~\bibnamefont {Kodate}},
  \bibinfo {author} {\bibfnamefont {H.}~\bibnamefont {Hibara}}, \bibinfo
  {author} {\bibfnamefont {Y.}~\bibnamefont {Kagaya}},\ and\ \bibinfo {author}
  {\bibfnamefont {K.}~\bibnamefont {Kinoshita}},\ }\bibfield  {title} {\bibinfo
  {title} {{{COXPRESdb}} v8: An animal gene coexpression database navigating
  from a global view to detailed investigations},\ }\href
  {https://doi.org/10.1093/nar/gkac983} {\bibfield  {journal} {\bibinfo
  {journal} {Nucleic Acids Research}\ }\textbf {\bibinfo {volume} {51}},\
  \bibinfo {pages} {D80} (\bibinfo {year} {2023})}\BibitemShut {NoStop}%
\bibitem [{\citenamefont {Huttenhower}\ \emph {et~al.}(2012)\citenamefont
  {Huttenhower}, \citenamefont {Gevers}, \citenamefont {Knight}, \citenamefont
  {Abubucker}, \citenamefont {Badger}, \citenamefont {Chinwalla}, \citenamefont
  {Creasy}, \citenamefont {Earl}, \citenamefont {FitzGerald}, \citenamefont
  {Fulton}, \citenamefont {Giglio}, \citenamefont {{Hallsworth-Pepin}},
  \citenamefont {Lobos}, \citenamefont {Madupu}, \citenamefont {Magrini},
  \citenamefont {Martin}, \citenamefont {Mitreva}, \citenamefont {Muzny},
  \citenamefont {Sodergren}, \citenamefont {Versalovic}, \citenamefont
  {Wollam}, \citenamefont {Worley}, \citenamefont {Wortman}, \citenamefont
  {Young}, \citenamefont {Zeng}, \citenamefont {Aagaard}, \citenamefont
  {Abolude}, \citenamefont {{Allen-Vercoe}}, \citenamefont {Alm}, \citenamefont
  {Alvarado}, \citenamefont {Andersen}, \citenamefont {Anderson}, \citenamefont
  {Appelbaum}, \citenamefont {Arachchi}, \citenamefont {Armitage},
  \citenamefont {Arze}, \citenamefont {Ayvaz}, \citenamefont {Baker},
  \citenamefont {Begg}, \citenamefont {Belachew}, \citenamefont {Bhonagiri},
  \citenamefont {Bihan}, \citenamefont {Blaser}, \citenamefont {Bloom},
  \citenamefont {Bonazzi}, \citenamefont {Paul~Brooks}, \citenamefont {Buck},
  \citenamefont {Buhay}, \citenamefont {Busam}, \citenamefont {Campbell},
  \citenamefont {Canon}, \citenamefont {Cantarel}, \citenamefont {Chain},
  \citenamefont {Chen}, \citenamefont {Chen}, \citenamefont {Chhibba},
  \citenamefont {Chu}, \citenamefont {Ciulla}, \citenamefont {Clemente},
  \citenamefont {Clifton}, \citenamefont {Conlan}, \citenamefont {Crabtree},
  \citenamefont {Cutting}, \citenamefont {Davidovics}, \citenamefont {Davis},
  \citenamefont {DeSantis}, \citenamefont {Deal}, \citenamefont {Delehaunty},
  \citenamefont {Dewhirst}, \citenamefont {Deych}, \citenamefont {Ding},
  \citenamefont {Dooling}, \citenamefont {Dugan}, \citenamefont
  {Michael~Dunne}, \citenamefont {Scott~Durkin}, \citenamefont {Edgar},
  \citenamefont {Erlich}, \citenamefont {Farmer}, \citenamefont {Farrell},
  \citenamefont {Faust}, \citenamefont {Feldgarden}, \citenamefont {Felix},
  \citenamefont {Fisher}, \citenamefont {Fodor}, \citenamefont {Forney},
  \citenamefont {Foster}, \citenamefont {Di~Francesco}, \citenamefont
  {Friedman}, \citenamefont {Friedrich}, \citenamefont {Fronick}, \citenamefont
  {Fulton}, \citenamefont {Gao}, \citenamefont {Garcia}, \citenamefont
  {Giannoukos}, \citenamefont {Giblin}, \citenamefont {Giovanni}, \citenamefont
  {Goldberg}, \citenamefont {Goll}, \citenamefont {Gonzalez}, \citenamefont
  {Griggs}, \citenamefont {Gujja}, \citenamefont {Kinder~Haake}, \citenamefont
  {Haas}, \citenamefont {Hamilton}, \citenamefont {Harris}, \citenamefont
  {Hepburn}, \citenamefont {Herter}, \citenamefont {Hoffmann}, \citenamefont
  {Holder}, \citenamefont {Howarth}, \citenamefont {Huang}, \citenamefont
  {Huse}, \citenamefont {Izard}, \citenamefont {Jansson}, \citenamefont
  {Jiang}, \citenamefont {Jordan}, \citenamefont {Joshi}, \citenamefont
  {Katancik}, \citenamefont {Keitel}, \citenamefont {Kelley}, \citenamefont
  {Kells}, \citenamefont {King}, \citenamefont {Knights}, \citenamefont {Kong},
  \citenamefont {Koren}, \citenamefont {Koren}, \citenamefont {Kota},
  \citenamefont {Kovar}, \citenamefont {Kyrpides}, \citenamefont {La~Rosa},
  \citenamefont {Lee}, \citenamefont {Lemon}, \citenamefont {Lennon},
  \citenamefont {Lewis}, \citenamefont {Lewis}, \citenamefont {Ley},
  \citenamefont {Li}, \citenamefont {Liolios}, \citenamefont {Liu},
  \citenamefont {Liu}, \citenamefont {Lo}, \citenamefont {Lozupone},
  \citenamefont {Dwayne~Lunsford}, \citenamefont {Madden}, \citenamefont
  {Mahurkar}, \citenamefont {Mannon}, \citenamefont {Mardis}, \citenamefont
  {Markowitz}, \citenamefont {Mavromatis}, \citenamefont {McCorrison},
  \citenamefont {McDonald}, \citenamefont {McEwen}, \citenamefont {McGuire},
  \citenamefont {McInnes}, \citenamefont {Mehta}, \citenamefont
  {Mihindukulasuriya}, \citenamefont {Miller}, \citenamefont {Minx},
  \citenamefont {Newsham}, \citenamefont {Nusbaum}, \citenamefont {O'Laughlin},
  \citenamefont {Orvis}, \citenamefont {Pagani}, \citenamefont {Palaniappan},
  \citenamefont {Patel}, \citenamefont {Pearson}, \citenamefont {Peterson},
  \citenamefont {Podar}, \citenamefont {Pohl}, \citenamefont {Pollard},
  \citenamefont {Pop}, \citenamefont {Priest}, \citenamefont {Proctor},
  \citenamefont {Qin}, \citenamefont {Raes}, \citenamefont {Ravel},
  \citenamefont {Reid}, \citenamefont {Rho}, \citenamefont {Rhodes},
  \citenamefont {Riehle}, \citenamefont {Rivera}, \citenamefont
  {{Rodriguez-Mueller}}, \citenamefont {Rogers}, \citenamefont {Ross},
  \citenamefont {Russ}, \citenamefont {Sanka}, \citenamefont {Sankar},
  \citenamefont {Fah~Sathirapongsasuti}, \citenamefont {Schloss}, \citenamefont
  {Schloss}, \citenamefont {Schmidt}, \citenamefont {Scholz}, \citenamefont
  {Schriml}, \citenamefont {Schubert}, \citenamefont {Segata}, \citenamefont
  {Segre}, \citenamefont {Shannon}, \citenamefont {Sharp}, \citenamefont
  {Sharpton}, \citenamefont {Shenoy}, \citenamefont {Sheth}, \citenamefont
  {Simone}, \citenamefont {Singh}, \citenamefont {Smillie}, \citenamefont
  {Sobel}, \citenamefont {Sommer}, \citenamefont {Spicer}, \citenamefont
  {Sutton}, \citenamefont {Sykes}, \citenamefont {Tabbaa}, \citenamefont
  {Thiagarajan}, \citenamefont {Tomlinson}, \citenamefont {Torralba},
  \citenamefont {Treangen}, \citenamefont {Truty}, \citenamefont
  {Vishnivetskaya}, \citenamefont {Walker}, \citenamefont {Wang}, \citenamefont
  {Wang}, \citenamefont {Ward}, \citenamefont {Warren}, \citenamefont {Watson},
  \citenamefont {Wellington}, \citenamefont {Wetterstrand}, \citenamefont
  {White}, \citenamefont {{Wilczek-Boney}}, \citenamefont {Wu}, \citenamefont
  {Wylie}, \citenamefont {Wylie}, \citenamefont {Yandava}, \citenamefont {Ye},
  \citenamefont {Ye}, \citenamefont {Yooseph}, \citenamefont {Youmans},
  \citenamefont {Zhang}, \citenamefont {Zhou}, \citenamefont {Zhu},
  \citenamefont {Zoloth}, \citenamefont {Zucker}, \citenamefont {Birren},
  \citenamefont {Gibbs}, \citenamefont {Highlander}, \citenamefont {Meth{\'e}},
  \citenamefont {Nelson}, \citenamefont {Petrosino}, \citenamefont {Weinstock},
  \citenamefont {Wilson}, \citenamefont {White},\ and\ \citenamefont {{The
  Human Microbiome Project Consortium}}}]{huttenhower_structure_2012}%
  \BibitemOpen
  \bibfield  {author} {\bibinfo {author} {\bibfnamefont {C.}~\bibnamefont
  {Huttenhower}}, \bibinfo {author} {\bibfnamefont {D.}~\bibnamefont {Gevers}},
  \bibinfo {author} {\bibfnamefont {R.}~\bibnamefont {Knight}}, \bibinfo
  {author} {\bibfnamefont {S.}~\bibnamefont {Abubucker}}, \bibinfo {author}
  {\bibfnamefont {J.~H.}\ \bibnamefont {Badger}}, \bibinfo {author}
  {\bibfnamefont {A.~T.}\ \bibnamefont {Chinwalla}}, \bibinfo {author}
  {\bibfnamefont {H.~H.}\ \bibnamefont {Creasy}}, \bibinfo {author}
  {\bibfnamefont {A.~M.}\ \bibnamefont {Earl}}, \bibinfo {author}
  {\bibfnamefont {M.~G.}\ \bibnamefont {FitzGerald}}, \bibinfo {author}
  {\bibfnamefont {R.~S.}\ \bibnamefont {Fulton}}, \bibinfo {author}
  {\bibfnamefont {M.~G.}\ \bibnamefont {Giglio}}, \bibinfo {author}
  {\bibfnamefont {K.}~\bibnamefont {{Hallsworth-Pepin}}}, \bibinfo {author}
  {\bibfnamefont {E.~A.}\ \bibnamefont {Lobos}}, \bibinfo {author}
  {\bibfnamefont {R.}~\bibnamefont {Madupu}}, \bibinfo {author} {\bibfnamefont
  {V.}~\bibnamefont {Magrini}}, \bibinfo {author} {\bibfnamefont {J.~C.}\
  \bibnamefont {Martin}}, \bibinfo {author} {\bibfnamefont {M.}~\bibnamefont
  {Mitreva}}, \bibinfo {author} {\bibfnamefont {D.~M.}\ \bibnamefont {Muzny}},
  \bibinfo {author} {\bibfnamefont {E.~J.}\ \bibnamefont {Sodergren}}, \bibinfo
  {author} {\bibfnamefont {J.}~\bibnamefont {Versalovic}}, \bibinfo {author}
  {\bibfnamefont {A.~M.}\ \bibnamefont {Wollam}}, \bibinfo {author}
  {\bibfnamefont {K.~C.}\ \bibnamefont {Worley}}, \bibinfo {author}
  {\bibfnamefont {J.~R.}\ \bibnamefont {Wortman}}, \bibinfo {author}
  {\bibfnamefont {S.~K.}\ \bibnamefont {Young}}, \bibinfo {author}
  {\bibfnamefont {Q.}~\bibnamefont {Zeng}}, \bibinfo {author} {\bibfnamefont
  {K.~M.}\ \bibnamefont {Aagaard}}, \bibinfo {author} {\bibfnamefont {O.~O.}\
  \bibnamefont {Abolude}}, \bibinfo {author} {\bibfnamefont {E.}~\bibnamefont
  {{Allen-Vercoe}}}, \bibinfo {author} {\bibfnamefont {E.~J.}\ \bibnamefont
  {Alm}}, \bibinfo {author} {\bibfnamefont {L.}~\bibnamefont {Alvarado}},
  \bibinfo {author} {\bibfnamefont {G.~L.}\ \bibnamefont {Andersen}}, \bibinfo
  {author} {\bibfnamefont {S.}~\bibnamefont {Anderson}}, \bibinfo {author}
  {\bibfnamefont {E.}~\bibnamefont {Appelbaum}}, \bibinfo {author}
  {\bibfnamefont {H.~M.}\ \bibnamefont {Arachchi}}, \bibinfo {author}
  {\bibfnamefont {G.}~\bibnamefont {Armitage}}, \bibinfo {author}
  {\bibfnamefont {C.~A.}\ \bibnamefont {Arze}}, \bibinfo {author}
  {\bibfnamefont {T.}~\bibnamefont {Ayvaz}}, \bibinfo {author} {\bibfnamefont
  {C.~C.}\ \bibnamefont {Baker}}, \bibinfo {author} {\bibfnamefont
  {L.}~\bibnamefont {Begg}}, \bibinfo {author} {\bibfnamefont {T.}~\bibnamefont
  {Belachew}}, \bibinfo {author} {\bibfnamefont {V.}~\bibnamefont {Bhonagiri}},
  \bibinfo {author} {\bibfnamefont {M.}~\bibnamefont {Bihan}}, \bibinfo
  {author} {\bibfnamefont {M.~J.}\ \bibnamefont {Blaser}}, \bibinfo {author}
  {\bibfnamefont {T.}~\bibnamefont {Bloom}}, \bibinfo {author} {\bibfnamefont
  {V.}~\bibnamefont {Bonazzi}}, \bibinfo {author} {\bibfnamefont
  {J.}~\bibnamefont {Paul~Brooks}}, \bibinfo {author} {\bibfnamefont {G.~A.}\
  \bibnamefont {Buck}}, \bibinfo {author} {\bibfnamefont {C.~J.}\ \bibnamefont
  {Buhay}}, \bibinfo {author} {\bibfnamefont {D.~A.}\ \bibnamefont {Busam}},
  \bibinfo {author} {\bibfnamefont {J.~L.}\ \bibnamefont {Campbell}}, \bibinfo
  {author} {\bibfnamefont {S.~R.}\ \bibnamefont {Canon}}, \bibinfo {author}
  {\bibfnamefont {B.~L.}\ \bibnamefont {Cantarel}}, \bibinfo {author}
  {\bibfnamefont {P.~S.~G.}\ \bibnamefont {Chain}}, \bibinfo {author}
  {\bibfnamefont {I.-M.~A.}\ \bibnamefont {Chen}}, \bibinfo {author}
  {\bibfnamefont {L.}~\bibnamefont {Chen}}, \bibinfo {author} {\bibfnamefont
  {S.}~\bibnamefont {Chhibba}}, \bibinfo {author} {\bibfnamefont
  {K.}~\bibnamefont {Chu}}, \bibinfo {author} {\bibfnamefont {D.~M.}\
  \bibnamefont {Ciulla}}, \bibinfo {author} {\bibfnamefont {J.~C.}\
  \bibnamefont {Clemente}}, \bibinfo {author} {\bibfnamefont {S.~W.}\
  \bibnamefont {Clifton}}, \bibinfo {author} {\bibfnamefont {S.}~\bibnamefont
  {Conlan}}, \bibinfo {author} {\bibfnamefont {J.}~\bibnamefont {Crabtree}},
  \bibinfo {author} {\bibfnamefont {M.~A.}\ \bibnamefont {Cutting}}, \bibinfo
  {author} {\bibfnamefont {N.~J.}\ \bibnamefont {Davidovics}}, \bibinfo
  {author} {\bibfnamefont {C.~C.}\ \bibnamefont {Davis}}, \bibinfo {author}
  {\bibfnamefont {T.~Z.}\ \bibnamefont {DeSantis}}, \bibinfo {author}
  {\bibfnamefont {C.}~\bibnamefont {Deal}}, \bibinfo {author} {\bibfnamefont
  {K.~D.}\ \bibnamefont {Delehaunty}}, \bibinfo {author} {\bibfnamefont
  {F.~E.}\ \bibnamefont {Dewhirst}}, \bibinfo {author} {\bibfnamefont
  {E.}~\bibnamefont {Deych}}, \bibinfo {author} {\bibfnamefont
  {Y.}~\bibnamefont {Ding}}, \bibinfo {author} {\bibfnamefont {D.~J.}\
  \bibnamefont {Dooling}}, \bibinfo {author} {\bibfnamefont {S.~P.}\
  \bibnamefont {Dugan}}, \bibinfo {author} {\bibfnamefont {W.}~\bibnamefont
  {Michael~Dunne}}, \bibinfo {author} {\bibfnamefont {A.}~\bibnamefont
  {Scott~Durkin}}, \bibinfo {author} {\bibfnamefont {R.~C.}\ \bibnamefont
  {Edgar}}, \bibinfo {author} {\bibfnamefont {R.~L.}\ \bibnamefont {Erlich}},
  \bibinfo {author} {\bibfnamefont {C.~N.}\ \bibnamefont {Farmer}}, \bibinfo
  {author} {\bibfnamefont {R.~M.}\ \bibnamefont {Farrell}}, \bibinfo {author}
  {\bibfnamefont {K.}~\bibnamefont {Faust}}, \bibinfo {author} {\bibfnamefont
  {M.}~\bibnamefont {Feldgarden}}, \bibinfo {author} {\bibfnamefont {V.~M.}\
  \bibnamefont {Felix}}, \bibinfo {author} {\bibfnamefont {S.}~\bibnamefont
  {Fisher}}, \bibinfo {author} {\bibfnamefont {A.~A.}\ \bibnamefont {Fodor}},
  \bibinfo {author} {\bibfnamefont {L.~J.}\ \bibnamefont {Forney}}, \bibinfo
  {author} {\bibfnamefont {L.}~\bibnamefont {Foster}}, \bibinfo {author}
  {\bibfnamefont {V.}~\bibnamefont {Di~Francesco}}, \bibinfo {author}
  {\bibfnamefont {J.}~\bibnamefont {Friedman}}, \bibinfo {author}
  {\bibfnamefont {D.~C.}\ \bibnamefont {Friedrich}}, \bibinfo {author}
  {\bibfnamefont {C.~C.}\ \bibnamefont {Fronick}}, \bibinfo {author}
  {\bibfnamefont {L.~L.}\ \bibnamefont {Fulton}}, \bibinfo {author}
  {\bibfnamefont {H.}~\bibnamefont {Gao}}, \bibinfo {author} {\bibfnamefont
  {N.}~\bibnamefont {Garcia}}, \bibinfo {author} {\bibfnamefont
  {G.}~\bibnamefont {Giannoukos}}, \bibinfo {author} {\bibfnamefont
  {C.}~\bibnamefont {Giblin}}, \bibinfo {author} {\bibfnamefont {M.~Y.}\
  \bibnamefont {Giovanni}}, \bibinfo {author} {\bibfnamefont {J.~M.}\
  \bibnamefont {Goldberg}}, \bibinfo {author} {\bibfnamefont {J.}~\bibnamefont
  {Goll}}, \bibinfo {author} {\bibfnamefont {A.}~\bibnamefont {Gonzalez}},
  \bibinfo {author} {\bibfnamefont {A.}~\bibnamefont {Griggs}}, \bibinfo
  {author} {\bibfnamefont {S.}~\bibnamefont {Gujja}}, \bibinfo {author}
  {\bibfnamefont {S.}~\bibnamefont {Kinder~Haake}}, \bibinfo {author}
  {\bibfnamefont {B.~J.}\ \bibnamefont {Haas}}, \bibinfo {author}
  {\bibfnamefont {H.~A.}\ \bibnamefont {Hamilton}}, \bibinfo {author}
  {\bibfnamefont {E.~L.}\ \bibnamefont {Harris}}, \bibinfo {author}
  {\bibfnamefont {T.~A.}\ \bibnamefont {Hepburn}}, \bibinfo {author}
  {\bibfnamefont {B.}~\bibnamefont {Herter}}, \bibinfo {author} {\bibfnamefont
  {D.~E.}\ \bibnamefont {Hoffmann}}, \bibinfo {author} {\bibfnamefont {M.~E.}\
  \bibnamefont {Holder}}, \bibinfo {author} {\bibfnamefont {C.}~\bibnamefont
  {Howarth}}, \bibinfo {author} {\bibfnamefont {K.~H.}\ \bibnamefont {Huang}},
  \bibinfo {author} {\bibfnamefont {S.~M.}\ \bibnamefont {Huse}}, \bibinfo
  {author} {\bibfnamefont {J.}~\bibnamefont {Izard}}, \bibinfo {author}
  {\bibfnamefont {J.~K.}\ \bibnamefont {Jansson}}, \bibinfo {author}
  {\bibfnamefont {H.}~\bibnamefont {Jiang}}, \bibinfo {author} {\bibfnamefont
  {C.}~\bibnamefont {Jordan}}, \bibinfo {author} {\bibfnamefont
  {V.}~\bibnamefont {Joshi}}, \bibinfo {author} {\bibfnamefont {J.~A.}\
  \bibnamefont {Katancik}}, \bibinfo {author} {\bibfnamefont {W.~A.}\
  \bibnamefont {Keitel}}, \bibinfo {author} {\bibfnamefont {S.~T.}\
  \bibnamefont {Kelley}}, \bibinfo {author} {\bibfnamefont {C.}~\bibnamefont
  {Kells}}, \bibinfo {author} {\bibfnamefont {N.~B.}\ \bibnamefont {King}},
  \bibinfo {author} {\bibfnamefont {D.}~\bibnamefont {Knights}}, \bibinfo
  {author} {\bibfnamefont {H.~H.}\ \bibnamefont {Kong}}, \bibinfo {author}
  {\bibfnamefont {O.}~\bibnamefont {Koren}}, \bibinfo {author} {\bibfnamefont
  {S.}~\bibnamefont {Koren}}, \bibinfo {author} {\bibfnamefont {K.~C.}\
  \bibnamefont {Kota}}, \bibinfo {author} {\bibfnamefont {C.~L.}\ \bibnamefont
  {Kovar}}, \bibinfo {author} {\bibfnamefont {N.~C.}\ \bibnamefont {Kyrpides}},
  \bibinfo {author} {\bibfnamefont {P.~S.}\ \bibnamefont {La~Rosa}}, \bibinfo
  {author} {\bibfnamefont {S.~L.}\ \bibnamefont {Lee}}, \bibinfo {author}
  {\bibfnamefont {K.~P.}\ \bibnamefont {Lemon}}, \bibinfo {author}
  {\bibfnamefont {N.}~\bibnamefont {Lennon}}, \bibinfo {author} {\bibfnamefont
  {C.~M.}\ \bibnamefont {Lewis}}, \bibinfo {author} {\bibfnamefont
  {L.}~\bibnamefont {Lewis}}, \bibinfo {author} {\bibfnamefont {R.~E.}\
  \bibnamefont {Ley}}, \bibinfo {author} {\bibfnamefont {K.}~\bibnamefont
  {Li}}, \bibinfo {author} {\bibfnamefont {K.}~\bibnamefont {Liolios}},
  \bibinfo {author} {\bibfnamefont {B.}~\bibnamefont {Liu}}, \bibinfo {author}
  {\bibfnamefont {Y.}~\bibnamefont {Liu}}, \bibinfo {author} {\bibfnamefont
  {C.-C.}\ \bibnamefont {Lo}}, \bibinfo {author} {\bibfnamefont {C.~A.}\
  \bibnamefont {Lozupone}}, \bibinfo {author} {\bibfnamefont {R.}~\bibnamefont
  {Dwayne~Lunsford}}, \bibinfo {author} {\bibfnamefont {T.}~\bibnamefont
  {Madden}}, \bibinfo {author} {\bibfnamefont {A.~A.}\ \bibnamefont
  {Mahurkar}}, \bibinfo {author} {\bibfnamefont {P.~J.}\ \bibnamefont
  {Mannon}}, \bibinfo {author} {\bibfnamefont {E.~R.}\ \bibnamefont {Mardis}},
  \bibinfo {author} {\bibfnamefont {V.~M.}\ \bibnamefont {Markowitz}}, \bibinfo
  {author} {\bibfnamefont {K.}~\bibnamefont {Mavromatis}}, \bibinfo {author}
  {\bibfnamefont {J.~M.}\ \bibnamefont {McCorrison}}, \bibinfo {author}
  {\bibfnamefont {D.}~\bibnamefont {McDonald}}, \bibinfo {author}
  {\bibfnamefont {J.}~\bibnamefont {McEwen}}, \bibinfo {author} {\bibfnamefont
  {A.~L.}\ \bibnamefont {McGuire}}, \bibinfo {author} {\bibfnamefont
  {P.}~\bibnamefont {McInnes}}, \bibinfo {author} {\bibfnamefont
  {T.}~\bibnamefont {Mehta}}, \bibinfo {author} {\bibfnamefont {K.~A.}\
  \bibnamefont {Mihindukulasuriya}}, \bibinfo {author} {\bibfnamefont {J.~R.}\
  \bibnamefont {Miller}}, \bibinfo {author} {\bibfnamefont {P.~J.}\
  \bibnamefont {Minx}}, \bibinfo {author} {\bibfnamefont {I.}~\bibnamefont
  {Newsham}}, \bibinfo {author} {\bibfnamefont {C.}~\bibnamefont {Nusbaum}},
  \bibinfo {author} {\bibfnamefont {M.}~\bibnamefont {O'Laughlin}}, \bibinfo
  {author} {\bibfnamefont {J.}~\bibnamefont {Orvis}}, \bibinfo {author}
  {\bibfnamefont {I.}~\bibnamefont {Pagani}}, \bibinfo {author} {\bibfnamefont
  {K.}~\bibnamefont {Palaniappan}}, \bibinfo {author} {\bibfnamefont {S.~M.}\
  \bibnamefont {Patel}}, \bibinfo {author} {\bibfnamefont {M.}~\bibnamefont
  {Pearson}}, \bibinfo {author} {\bibfnamefont {J.}~\bibnamefont {Peterson}},
  \bibinfo {author} {\bibfnamefont {M.}~\bibnamefont {Podar}}, \bibinfo
  {author} {\bibfnamefont {C.}~\bibnamefont {Pohl}}, \bibinfo {author}
  {\bibfnamefont {K.~S.}\ \bibnamefont {Pollard}}, \bibinfo {author}
  {\bibfnamefont {M.}~\bibnamefont {Pop}}, \bibinfo {author} {\bibfnamefont
  {M.~E.}\ \bibnamefont {Priest}}, \bibinfo {author} {\bibfnamefont {L.~M.}\
  \bibnamefont {Proctor}}, \bibinfo {author} {\bibfnamefont {X.}~\bibnamefont
  {Qin}}, \bibinfo {author} {\bibfnamefont {J.}~\bibnamefont {Raes}}, \bibinfo
  {author} {\bibfnamefont {J.}~\bibnamefont {Ravel}}, \bibinfo {author}
  {\bibfnamefont {J.~G.}\ \bibnamefont {Reid}}, \bibinfo {author}
  {\bibfnamefont {M.}~\bibnamefont {Rho}}, \bibinfo {author} {\bibfnamefont
  {R.}~\bibnamefont {Rhodes}}, \bibinfo {author} {\bibfnamefont {K.~P.}\
  \bibnamefont {Riehle}}, \bibinfo {author} {\bibfnamefont {M.~C.}\
  \bibnamefont {Rivera}}, \bibinfo {author} {\bibfnamefont {B.}~\bibnamefont
  {{Rodriguez-Mueller}}}, \bibinfo {author} {\bibfnamefont {Y.-H.}\
  \bibnamefont {Rogers}}, \bibinfo {author} {\bibfnamefont {M.~C.}\
  \bibnamefont {Ross}}, \bibinfo {author} {\bibfnamefont {C.}~\bibnamefont
  {Russ}}, \bibinfo {author} {\bibfnamefont {R.~K.}\ \bibnamefont {Sanka}},
  \bibinfo {author} {\bibfnamefont {P.}~\bibnamefont {Sankar}}, \bibinfo
  {author} {\bibfnamefont {J.}~\bibnamefont {Fah~Sathirapongsasuti}}, \bibinfo
  {author} {\bibfnamefont {J.~A.}\ \bibnamefont {Schloss}}, \bibinfo {author}
  {\bibfnamefont {P.~D.}\ \bibnamefont {Schloss}}, \bibinfo {author}
  {\bibfnamefont {T.~M.}\ \bibnamefont {Schmidt}}, \bibinfo {author}
  {\bibfnamefont {M.}~\bibnamefont {Scholz}}, \bibinfo {author} {\bibfnamefont
  {L.}~\bibnamefont {Schriml}}, \bibinfo {author} {\bibfnamefont {A.~M.}\
  \bibnamefont {Schubert}}, \bibinfo {author} {\bibfnamefont {N.}~\bibnamefont
  {Segata}}, \bibinfo {author} {\bibfnamefont {J.~A.}\ \bibnamefont {Segre}},
  \bibinfo {author} {\bibfnamefont {W.~D.}\ \bibnamefont {Shannon}}, \bibinfo
  {author} {\bibfnamefont {R.~R.}\ \bibnamefont {Sharp}}, \bibinfo {author}
  {\bibfnamefont {T.~J.}\ \bibnamefont {Sharpton}}, \bibinfo {author}
  {\bibfnamefont {N.}~\bibnamefont {Shenoy}}, \bibinfo {author} {\bibfnamefont
  {N.~U.}\ \bibnamefont {Sheth}}, \bibinfo {author} {\bibfnamefont {G.~A.}\
  \bibnamefont {Simone}}, \bibinfo {author} {\bibfnamefont {I.}~\bibnamefont
  {Singh}}, \bibinfo {author} {\bibfnamefont {C.~S.}\ \bibnamefont {Smillie}},
  \bibinfo {author} {\bibfnamefont {J.~D.}\ \bibnamefont {Sobel}}, \bibinfo
  {author} {\bibfnamefont {D.~D.}\ \bibnamefont {Sommer}}, \bibinfo {author}
  {\bibfnamefont {P.}~\bibnamefont {Spicer}}, \bibinfo {author} {\bibfnamefont
  {G.~G.}\ \bibnamefont {Sutton}}, \bibinfo {author} {\bibfnamefont {S.~M.}\
  \bibnamefont {Sykes}}, \bibinfo {author} {\bibfnamefont {D.~G.}\ \bibnamefont
  {Tabbaa}}, \bibinfo {author} {\bibfnamefont {M.}~\bibnamefont {Thiagarajan}},
  \bibinfo {author} {\bibfnamefont {C.~M.}\ \bibnamefont {Tomlinson}}, \bibinfo
  {author} {\bibfnamefont {M.}~\bibnamefont {Torralba}}, \bibinfo {author}
  {\bibfnamefont {T.~J.}\ \bibnamefont {Treangen}}, \bibinfo {author}
  {\bibfnamefont {R.~M.}\ \bibnamefont {Truty}}, \bibinfo {author}
  {\bibfnamefont {T.~A.}\ \bibnamefont {Vishnivetskaya}}, \bibinfo {author}
  {\bibfnamefont {J.}~\bibnamefont {Walker}}, \bibinfo {author} {\bibfnamefont
  {L.}~\bibnamefont {Wang}}, \bibinfo {author} {\bibfnamefont {Z.}~\bibnamefont
  {Wang}}, \bibinfo {author} {\bibfnamefont {D.~V.}\ \bibnamefont {Ward}},
  \bibinfo {author} {\bibfnamefont {W.}~\bibnamefont {Warren}}, \bibinfo
  {author} {\bibfnamefont {M.~A.}\ \bibnamefont {Watson}}, \bibinfo {author}
  {\bibfnamefont {C.}~\bibnamefont {Wellington}}, \bibinfo {author}
  {\bibfnamefont {K.~A.}\ \bibnamefont {Wetterstrand}}, \bibinfo {author}
  {\bibfnamefont {J.~R.}\ \bibnamefont {White}}, \bibinfo {author}
  {\bibfnamefont {K.}~\bibnamefont {{Wilczek-Boney}}}, \bibinfo {author}
  {\bibfnamefont {Y.}~\bibnamefont {Wu}}, \bibinfo {author} {\bibfnamefont
  {K.~M.}\ \bibnamefont {Wylie}}, \bibinfo {author} {\bibfnamefont
  {T.}~\bibnamefont {Wylie}}, \bibinfo {author} {\bibfnamefont
  {C.}~\bibnamefont {Yandava}}, \bibinfo {author} {\bibfnamefont
  {L.}~\bibnamefont {Ye}}, \bibinfo {author} {\bibfnamefont {Y.}~\bibnamefont
  {Ye}}, \bibinfo {author} {\bibfnamefont {S.}~\bibnamefont {Yooseph}},
  \bibinfo {author} {\bibfnamefont {B.~P.}\ \bibnamefont {Youmans}}, \bibinfo
  {author} {\bibfnamefont {L.}~\bibnamefont {Zhang}}, \bibinfo {author}
  {\bibfnamefont {Y.}~\bibnamefont {Zhou}}, \bibinfo {author} {\bibfnamefont
  {Y.}~\bibnamefont {Zhu}}, \bibinfo {author} {\bibfnamefont {L.}~\bibnamefont
  {Zoloth}}, \bibinfo {author} {\bibfnamefont {J.~D.}\ \bibnamefont {Zucker}},
  \bibinfo {author} {\bibfnamefont {B.~W.}\ \bibnamefont {Birren}}, \bibinfo
  {author} {\bibfnamefont {R.~A.}\ \bibnamefont {Gibbs}}, \bibinfo {author}
  {\bibfnamefont {S.~K.}\ \bibnamefont {Highlander}}, \bibinfo {author}
  {\bibfnamefont {B.~A.}\ \bibnamefont {Meth{\'e}}}, \bibinfo {author}
  {\bibfnamefont {K.~E.}\ \bibnamefont {Nelson}}, \bibinfo {author}
  {\bibfnamefont {J.~F.}\ \bibnamefont {Petrosino}}, \bibinfo {author}
  {\bibfnamefont {G.~M.}\ \bibnamefont {Weinstock}}, \bibinfo {author}
  {\bibfnamefont {R.~K.}\ \bibnamefont {Wilson}}, \bibinfo {author}
  {\bibfnamefont {O.}~\bibnamefont {White}},\ and\ \bibinfo {author}
  {\bibnamefont {{The Human Microbiome Project Consortium}}},\ }\bibfield
  {title} {\bibinfo {title} {Structure, function and diversity of the healthy
  human microbiome},\ }\href {https://doi.org/10.1038/nature11234} {\bibfield
  {journal} {\bibinfo  {journal} {Nature}\ }\textbf {\bibinfo {volume} {486}},\
  \bibinfo {pages} {207} (\bibinfo {year} {2012})}\BibitemShut {NoStop}%
\bibitem [{\citenamefont {Thompson}\ \emph {et~al.}(2017)\citenamefont
  {Thompson}, \citenamefont {Sanders}, \citenamefont {McDonald}, \citenamefont
  {Amir}, \citenamefont {Ladau}, \citenamefont {Locey}, \citenamefont {Prill},
  \citenamefont {Tripathi}, \citenamefont {Gibbons}, \citenamefont {Ackermann},
  \citenamefont {{Navas-Molina}}, \citenamefont {Janssen}, \citenamefont
  {Kopylova}, \citenamefont {{V{\'a}zquez-Baeza}}, \citenamefont
  {Gonz{\'a}lez}, \citenamefont {Morton}, \citenamefont {Mirarab},
  \citenamefont {Zech~Xu}, \citenamefont {Jiang}, \citenamefont {Haroon},
  \citenamefont {Kanbar}, \citenamefont {Zhu}, \citenamefont {Jin~Song},
  \citenamefont {Kosciolek}, \citenamefont {Bokulich}, \citenamefont {Lefler},
  \citenamefont {Brislawn}, \citenamefont {Humphrey}, \citenamefont {Owens},
  \citenamefont {{Hampton-Marcell}}, \citenamefont {{Berg-Lyons}},
  \citenamefont {McKenzie}, \citenamefont {Fierer}, \citenamefont {Fuhrman},
  \citenamefont {Clauset}, \citenamefont {Stevens}, \citenamefont {Shade},
  \citenamefont {Pollard}, \citenamefont {Goodwin}, \citenamefont {Jansson},
  \citenamefont {Gilbert},\ and\ \citenamefont
  {Knight}}]{thompson_communal_2017}%
  \BibitemOpen
  \bibfield  {author} {\bibinfo {author} {\bibfnamefont {L.~R.}\ \bibnamefont
  {Thompson}}, \bibinfo {author} {\bibfnamefont {J.~G.}\ \bibnamefont
  {Sanders}}, \bibinfo {author} {\bibfnamefont {D.}~\bibnamefont {McDonald}},
  \bibinfo {author} {\bibfnamefont {A.}~\bibnamefont {Amir}}, \bibinfo {author}
  {\bibfnamefont {J.}~\bibnamefont {Ladau}}, \bibinfo {author} {\bibfnamefont
  {K.~J.}\ \bibnamefont {Locey}}, \bibinfo {author} {\bibfnamefont {R.~J.}\
  \bibnamefont {Prill}}, \bibinfo {author} {\bibfnamefont {A.}~\bibnamefont
  {Tripathi}}, \bibinfo {author} {\bibfnamefont {S.~M.}\ \bibnamefont
  {Gibbons}}, \bibinfo {author} {\bibfnamefont {G.}~\bibnamefont {Ackermann}},
  \bibinfo {author} {\bibfnamefont {J.~A.}\ \bibnamefont {{Navas-Molina}}},
  \bibinfo {author} {\bibfnamefont {S.}~\bibnamefont {Janssen}}, \bibinfo
  {author} {\bibfnamefont {E.}~\bibnamefont {Kopylova}}, \bibinfo {author}
  {\bibfnamefont {Y.}~\bibnamefont {{V{\'a}zquez-Baeza}}}, \bibinfo {author}
  {\bibfnamefont {A.}~\bibnamefont {Gonz{\'a}lez}}, \bibinfo {author}
  {\bibfnamefont {J.~T.}\ \bibnamefont {Morton}}, \bibinfo {author}
  {\bibfnamefont {S.}~\bibnamefont {Mirarab}}, \bibinfo {author} {\bibfnamefont
  {Z.}~\bibnamefont {Zech~Xu}}, \bibinfo {author} {\bibfnamefont
  {L.}~\bibnamefont {Jiang}}, \bibinfo {author} {\bibfnamefont {M.~F.}\
  \bibnamefont {Haroon}}, \bibinfo {author} {\bibfnamefont {J.}~\bibnamefont
  {Kanbar}}, \bibinfo {author} {\bibfnamefont {Q.}~\bibnamefont {Zhu}},
  \bibinfo {author} {\bibfnamefont {S.}~\bibnamefont {Jin~Song}}, \bibinfo
  {author} {\bibfnamefont {T.}~\bibnamefont {Kosciolek}}, \bibinfo {author}
  {\bibfnamefont {N.~A.}\ \bibnamefont {Bokulich}}, \bibinfo {author}
  {\bibfnamefont {J.}~\bibnamefont {Lefler}}, \bibinfo {author} {\bibfnamefont
  {C.~J.}\ \bibnamefont {Brislawn}}, \bibinfo {author} {\bibfnamefont
  {G.}~\bibnamefont {Humphrey}}, \bibinfo {author} {\bibfnamefont {S.~M.}\
  \bibnamefont {Owens}}, \bibinfo {author} {\bibfnamefont {J.}~\bibnamefont
  {{Hampton-Marcell}}}, \bibinfo {author} {\bibfnamefont {D.}~\bibnamefont
  {{Berg-Lyons}}}, \bibinfo {author} {\bibfnamefont {V.}~\bibnamefont
  {McKenzie}}, \bibinfo {author} {\bibfnamefont {N.}~\bibnamefont {Fierer}},
  \bibinfo {author} {\bibfnamefont {J.~A.}\ \bibnamefont {Fuhrman}}, \bibinfo
  {author} {\bibfnamefont {A.}~\bibnamefont {Clauset}}, \bibinfo {author}
  {\bibfnamefont {R.~L.}\ \bibnamefont {Stevens}}, \bibinfo {author}
  {\bibfnamefont {A.}~\bibnamefont {Shade}}, \bibinfo {author} {\bibfnamefont
  {K.~S.}\ \bibnamefont {Pollard}}, \bibinfo {author} {\bibfnamefont {K.~D.}\
  \bibnamefont {Goodwin}}, \bibinfo {author} {\bibfnamefont {J.~K.}\
  \bibnamefont {Jansson}}, \bibinfo {author} {\bibfnamefont {J.~A.}\
  \bibnamefont {Gilbert}},\ and\ \bibinfo {author} {\bibfnamefont
  {R.}~\bibnamefont {Knight}},\ }\bibfield  {title} {\bibinfo {title} {A
  communal catalogue reveals {{Earth}}'s multiscale microbial diversity},\
  }\href {https://doi.org/10.1038/nature24621} {\bibfield  {journal} {\bibinfo
  {journal} {Nature}\ }\textbf {\bibinfo {volume} {551}},\ \bibinfo {pages}
  {457} (\bibinfo {year} {2017})}\BibitemShut {NoStop}%
\bibitem [{\citenamefont {Johnson}\ \emph {et~al.}(2007)\citenamefont
  {Johnson}, \citenamefont {Li},\ and\ \citenamefont
  {Rabinovic}}]{johnson_adjusting_2007}%
  \BibitemOpen
  \bibfield  {author} {\bibinfo {author} {\bibfnamefont {W.~E.}\ \bibnamefont
  {Johnson}}, \bibinfo {author} {\bibfnamefont {C.}~\bibnamefont {Li}},\ and\
  \bibinfo {author} {\bibfnamefont {A.}~\bibnamefont {Rabinovic}},\ }\bibfield
  {title} {\bibinfo {title} {Adjusting batch effects in microarray expression
  data using empirical {{Bayes}} methods},\ }\href
  {https://doi.org/10.1093/biostatistics/kxj037} {\bibfield  {journal}
  {\bibinfo  {journal} {Biostatistics}\ }\textbf {\bibinfo {volume} {8}},\
  \bibinfo {pages} {118} (\bibinfo {year} {2007})}\BibitemShut {NoStop}%
\bibitem [{\citenamefont {Besag}(1974)}]{besag_spatial_1974}%
  \BibitemOpen
  \bibfield  {author} {\bibinfo {author} {\bibfnamefont {J.}~\bibnamefont
  {Besag}},\ }\bibfield  {title} {\bibinfo {title} {Spatial {{Interaction}} and
  the {{Statistical Analysis}} of {{Lattice Systems}}},\ }\href
  {https://doi.org/10.1111/j.2517-6161.1974.tb00999.x} {\bibfield  {journal}
  {\bibinfo  {journal} {Journal of the Royal Statistical Society: Series B
  (Methodological)}\ }\textbf {\bibinfo {volume} {36}},\ \bibinfo {pages} {192}
  (\bibinfo {year} {1974})}\BibitemShut {NoStop}%
\bibitem [{\citenamefont {Khare}\ \emph {et~al.}(2015)\citenamefont {Khare},
  \citenamefont {Oh},\ and\ \citenamefont {Rajaratnam}}]{khare_convex_2015}%
  \BibitemOpen
  \bibfield  {author} {\bibinfo {author} {\bibfnamefont {K.}~\bibnamefont
  {Khare}}, \bibinfo {author} {\bibfnamefont {S.-Y.}\ \bibnamefont {Oh}},\ and\
  \bibinfo {author} {\bibfnamefont {B.}~\bibnamefont {Rajaratnam}},\ }\bibfield
   {title} {\bibinfo {title} {A {{Convex Pseudolikelihood Framework}} for
  {{High Dimensional Partial Correlation Estimation}} with {{Convergence
  Guarantees}}},\ }\href {https://doi.org/10.1111/rssb.12088} {\bibfield
  {journal} {\bibinfo  {journal} {Journal of the Royal Statistical Society
  Series B: Statistical Methodology}\ }\textbf {\bibinfo {volume} {77}},\
  \bibinfo {pages} {803} (\bibinfo {year} {2015})}\BibitemShut {NoStop}%
\end{thebibliography}%

\appendix

\section{Generative models}\label{app:models}

In our examples we use two graphical models: the Ising
model~\cite{nguyen_inverse_2017} and a multivariate Gaussian
distribution~\cite{dempster_covariance_1972}. The Ising model is a
distribution on $N$ binary variables $\bm{x} \in \{-1,1\}^{N}$ given by
\begin{equation}
  P(\bm x | \bm W, \bm \theta) = \frac{\ee^{\sum_{i<j}W_{ij}x_{i}x_{j} + \sum_{i}\theta_{i}x_{i}}}{Z(\W, \bm \theta)},
\end{equation}
with $\theta_{i}$ being a local field on node $i$, and
$Z(\W,\bm\theta)=\sum_{\bm x}\ee^{\sum_{i<j}W_{ij}x_{i}x_{j} + \sum_{i}\theta_{i}x_{i}}$
a normalization constant. Since this normalization cannot be computed in closed
form, we make use of the pseudolikelihood
approximation~\cite{besag_spatial_1974},
\begin{align}
  P(\bm x | \bm W, \bm \theta) &= \prod_{i}P(x_{i}|\bm x\setminus {x_{i}}, \bm W, \bm \theta)\\
  &= \prod_{i}\frac{\ee^{x_{i}(\sum_{j}W_{ij}x_{j} + \theta_i)}}{2\cosh(\sum_{j}W_{ij}x_{j} + \theta_i)},
\end{align}
as it gives asymptotically correct results and has excellent performance in
practice~\cite{nguyen_inverse_2017}. Likewise, the (zero-mean) multivariate
Gaussian is a distribution on $\bm{x} \in \mathbb{R}^{N}$ given by
\begin{equation}
  P(\bm x | \bm W) = \frac{\ee^{-\frac{1}{2} {\bm x}^{\top}\W \bm x}}{\sqrt{(2\pi)^N |\bm W^{-1}|}},
\end{equation}
where $\bm W$ is the precision (or inverse covariance) matrix. Unlike the Ising
model, this likelihood is analytical --- nevertheless, the evaluation of the
determinant is computationally expensive, and therefore we make use of the same
pseudolikelihood approximation~\cite{khare_convex_2015},
\begin{equation}
  P(\bm x | \bm W, \bm \theta) = \prod_{i}\frac{\ee^{-(x_i + \theta_i^2\sum_{j\neq i}W_{ij}x_{j})^{2}/2\theta_i^2}}{\sqrt{(2\pi)}\theta_i},
\end{equation}
where we parameterize the diagonal entries as $\theta_i=1/\sqrt{W_{ii}}$.

In both cases, we have an additional set of $N$ parameters $\bm \theta$ which
we update alongside the matrix $\bm W$ in our algorithms. Updates on an
individual entry $W_{ij}$ of $\W$ can be computed in time $O(1)$ (independently
of the degrees of $i$ and $j$ in a sparse representation of $\bm W$) by keeping
track of the weighted sum of the neighbors $m_{i}=\sum_{j\neq i}W_{ij}x_{j}$ for
every node and updating it as appropriate.

For both models we use a Laplace prior
\begin{equation}
  P(\bm W|\lambda) = \prod_{i<j} \lambda\ee^{-\lambda |W_{ij}|}/2,
\end{equation}
which provides a convex $L_{1}$ regularization with a penalty given by
$\lambda$, chosen to achieve a desired level of sparsity.

\section{Low-level optimizations}\label{app:optimizations}

Below we describe a few low-level optimizations that we found to give good
improvements to the runtime of the algorithm we propose in the main text.

\paragraph{Caching.}
The typical case for objective functions $\pi(\W)$ is that the computation of
$\max_{W_{ij}}\pi(\W)$ will require $O(M)$ operations, where $M$ is the number
of data samples available. Since this computation is done in the innermost loops
of algorithm~\ref{alg:knn}, it will amount to an overall multiplicative factor
of $O(M)$ in its runtime. However, because the distance function
$\textsc{d}(i,j)$ will be called multiple times for the same pair $(i,j)$, a
good optimization strategy is to cache its values, for example in a hash table,
such that repeated calls will take time $O(1)$ rather than $O(M)$. We find that
this optimization can reduce the total runtime of algorithm~\ref{alg:knn} by at
least one order of magnitude in typical cases.

\paragraph{Gradient as distance.}
The definition of $\textsc{d}(i,j)=-\max_{W_{ij}}\pi(\W)$ is sufficient to
guarantee the correctness of the algorithm~\ref{alg:mclosest}, but in situations
where it cannot be computed in closed form, requiring for example a bisection
search, a faster approach is to use instead the absolute gradient
$\textsc{d}(i,j)=-|\frac{\partial}{\partial W_{ij}}\log\pi(\W)|$, which often
can be analytically computed or well approximated with finite difference. In
general this requires substantially fewer likelihood evaluations than bisection
search. This approach is strictly applicable only with differentiable
objectives, although we observed correct behavior for $L_{1}$-regularized
likelihoods when approximating the gradient using central finite difference. We
observed that this optimization improves the runtime by a factor of around six
in typical scenarios.

\paragraph{Parallelization.}
The workhorse of the algorithm is the NNDescent search (algorithm
~\ref{alg:knn}), which is easily parallelizable in a shared memory environment,
since the neighborhood of each node can be inspected, and its list of nearest
neighbors can be updated, in a manner that is completely independent from the
other nodes, and hence requires no synchronization. Thus, the parallel execution
of algorithm~\ref{alg:knn} is straightforward.

The actual updates of the matrix $\bm W$ in the GCD algorithm~\ref{alg:gcd} can
also be done in parallel, but that requires some synchronization. For many
objectives $\pi(\W)$, we can only consider the change of one value $W_{ij}$ at a
time for each node $i$ and $j$, since the likelihood will involve sums of the
type $m_{i}=\sum_{j}W_{ij}x_{j}$ for a node $i$, where $x_{j}$ are data values.
Therefore, only the subset of the edges $\mathcal E_{\text{best}}$ in
algorithm~\ref{alg:gcd} that are incident to an \emph{independent vertex set} in
the graph will be able to be updated in parallel. This can be implemented with
mutexes on each node, which are simultaneously locked by each thread (without
blocking) for each pair $(i,j)$ before $W_{ij}$ is updated, which otherwise
proceeds to the next pair if the lock cannot be acquired. Empirically, we found
that is enough to keep up to 256 threads busy with little contention for
$N>10^{4}$ with our OpenMP implementation~\cite{peixoto_graph-tool_2014}.

\end{document}